\documentclass[12pt]{article}
\usepackage{amssymb, amsmath}
\usepackage{amsfonts}
\usepackage{color}
\usepackage{mathrsfs}
\usepackage{cite}

\newcommand{\ret}{\nonumber \\}

\newcommand{\Section}[1]%
{\section{#1}\setcounter{equation}{0}%
\setcounter{theorem}{0}}
\newtheorem{theorem}{Theorem}
\newtheorem{lemma}[theorem]{Lemma}
\newtheorem{coro}[theorem]{Corollary}
\newtheorem{pro}[theorem]{Proposition}

\newenvironment{proof}[1]%
{\par\noindent{\em #1:\ }}%
{~\rule{2mm}{2mm}\par\bigskip}
\def\re{\mathbb{R}}

\def\ze{\mathbb{Z}}

\begin{document}
\newpage\thispagestyle{empty}
{\topskip 2cm
\begin{center}
{\Large\bf Nambu-Goldstone Modes\\  
\bigskip
for Superconducting Lattice Fermions\\}
\bigskip\bigskip
{\Large Tohru Koma\\}
\bigskip\medskip
{\small Department of Physics, Gakushuin University (retired), 
Mejiro, Toshima-ku, Tokyo 171-8588, JAPAN\\}
\smallskip

\end{center}
\vfil
\noindent
{\bf Abstract:} \\
We present a lattice model for superconducting fermions whose nearest-neighbour two-body interactions 
are a Bardeen-Cooper-Schrieffer-type pairing on the hypercubic lattice $\ze^d$ with the dimension $d\ge 3$. 
Although these effective interactions between two electrons are believed to be caused by electron-phonon interactions, 
we assume that the interactions between two electrons are of short range without phonons.
For the model, we prove the existence of the long-range order of the superconductivity at low temperatures, 
and also prove the existence of a gapless excitation above the infinite-volume ground state. 
Namely, there appears a Nambu-Goldstone mode which is associated with the U(1) symmetry breaking. 
The corresponding Nambu-Goldstone boson is exactly a Cooper pair, which consists of spin up and down electrons. 
\par\noindent
\bigskip
\hrule
\bigskip


\vfil}\newpage
\tableofcontents
\newpage
\Section{Introduction}

The aim of this paper is to prove the existence of phase transitions at non-zero low temperatures 
for interacting fermion systems. In particular, we focus on superconductivity phase transitions, 
which have been often described by using Bardeen-Cooper-Schrieffer(BCS) theory \cite{BCS}  
as a microscopic effect caused by a condensation of Cooper pairs.
Within the BCS theory, the electron-phonon interaction induces an effective interaction 
that binds spin up and down electrons into a paired state. 

In the present paper, we assume that the effective interactions between two electrons are of short range without phonons. 
More precisely, those are assumed to be a nearest-neighbour two-body interaction 
which is a Bardeen-Cooper-Schrieffer-type pairing \cite{BCS} on the hypercubic lattice $\ze^d$ with 
the dimensions $d\ge 3$. Namely, the electrons hop on the lattice $\ze^d$, being affected by the BCS-type pairing interactions. 
For this lattice model, we prove the existence of the long-range order of the superconductivity at low temperatures, 
and also prove the existence of a gapless excitation above the infinite-volume ground state. 
Namely, there appears a Nambu-Goldstone mode \cite{Nambu,NJL,Goldstone,GSW} which is associated with the U(1) symmetry breaking. 
The corresponding Nambu-Goldstone boson is exactly a Cooper pair, which consists of spin up and down electrons. 
Thus, our results justify the BCS picture for superconductivity, although the excitations exhibit a gapless spectrum 
above the symmetry-breaking ground state.

In order to prove these statements, we basically rely on the method of reflection positivity \cite{FSS,DLS,FILS1,FILS,KLS,KLS2,JP}. 
However, we have to deal with the reflection positivity for fermion systems.  
Although there have been many applications of the reflection positivity to various fermion systems so far, 
it is not so easy to realize the reflection positivity for a fermion system. 
The reflection positivity on the spin degrees of freedom of fermions was initiated by Lieb \cite{Lieb}, 
and has been applied to many fermion systems \cite{KuboKishi,UTS,Tian1,ShenQiu,SQT,YS,Shen,Shen2,Tian,YoshidaKatsura,FL,Miyao}. 
On the other hand, the application of the usual real space reflection positivity is restricted to 
the flux phase problem \cite{LiebFlux,MN} and one-dimensional systems \cite{LN,CJLP,MOT}. 
Further, by relying on the properties of Majorana fermions, the reflection positivity yields some results 
\cite{JP2,WHXX,YoshidaKatsura}.  

In the present paper, we realize the real space reflection positivity for the fermion system mentioned above with 
a certain hopping Hamiltonian on the hypercubic lattice $\ze^d$ with the dimension $d\ge 3$ by relying 
the Majorana representation \cite{JP} of the fermions. 

\bigskip\bigskip

\noindent
{\bf Acknowledgements:} I would like to thank Yasuhiro Tada for valuable discussions. 
The present paper is an answer to his question whether the methods developed for 
the quantum antiferromagnets \cite{KomaTasaki1,KomaTasaki2,Koma1,Koma2,Koma3} are applicable to superconducting fermion models.

\Section{Model and main results}

We first describe our model and the precise statements of our main results. 

Consider a $d$-dimensional finite hypercubic lattice, 
\begin{equation}
\label{Lambda}
\Lambda:=\{-L+1,-L+2,\ldots,-1,0,1,\ldots,L-1,L\}^d\subset\mathbb{Z}^d,
\end{equation}
with a positive integer $L$ and $d\ge 1$. For each lattice site $x\in\Lambda$, 
we introduce a fermion operator $a_{x,\sigma}$ with spin degrees of freedom $\sigma=\uparrow,\downarrow$. 
The fermion operators satisfy the anticommutation relations, 
\begin{equation}
\{a_{x,\sigma},a_{x',\sigma'}^\dagger\}=\delta_{x,x'}\delta_{\sigma,\sigma'}\quad 
\mbox{and}\quad \{a_{x,\sigma},a_{x',\sigma'}\}=0,
\end{equation}
for $x,x'\in\Lambda$ and $\sigma,\sigma'=\uparrow,\downarrow$. The number operator is given by 
\begin{equation}
n_{x,\sigma}:=a_{x,\sigma}^\dagger a_{x,\sigma}
\end{equation}
for $x\in\Lambda$ and $\sigma=\uparrow,\downarrow$. 

The Hamiltonian which we consider is given by 
\begin{equation}
\label{HamB}
H^{(\Lambda)}(B)=H_{\rm hop}^{(\Lambda)}+H_{\rm int}^{(\Lambda)}-BO^{(\Lambda)}
\end{equation}
which consists of three terms, hopping, two-body-interaction and symmetry-breaking-field Hamiltonians. 
The order parameter $O^{(\Lambda)}$ of superconductivity in the third term is given by 
\begin{equation}
\label{O}
O^{(\Lambda)}:=\sum_{x\in\Lambda}(-1)^{x^{(1)}+\cdots+x^{(d)}}
i(a_{x,\uparrow}^\dagger a_{x,\downarrow}^\dagger-a_{x,\downarrow}a_{x,\uparrow}), 
\end{equation}
where we have written $x=(x^{(1)},x^{(2)},\ldots,x^{(d)})$ for the $d$-dimensional coordinates, and $B\in\re$ is 
the symmetry breaking field. 
The interaction Hamiltonian $H_{\rm int}^{(\Lambda)}$ of the second term is given by  
\begin{eqnarray}
\label{Hint}   
H_{\rm int}^{(\Lambda)}&=&g\sum_{\{x,y\}\subset\Lambda:|x-y|=1}(a_{x,\uparrow}^\dagger a_{x,\downarrow}^\dagger 
a_{y,\downarrow}a_{y,\uparrow}+a_{y,\uparrow}^\dagger a_{y,\downarrow}^\dagger a_{x,\downarrow}a_{x,\uparrow})\ret
&+&g'\sum_{\{x,y\}\subset\Lambda:|x-y|=1}
(n_{x,\uparrow}+n_{x,\downarrow}-1)(n_{y,\uparrow}+n_{y,\downarrow}-1) 
\end{eqnarray}
with the coupling constants, $g>0$ and $g'\ge 0$. The summand of the first sum in the right-hand side is  
a nearest-neighbour two-body interaction which is a Bardeen-Cooper-Schrieffer-type pairing \cite{BCS}. 
The second one is also a nearest-neighbour two-body interaction which is a repulsive Coulomb interaction. 
For the interaction Hamiltonian, we impose the periodic boundary conditions in all the directions of the lattice $\Lambda$. 

The hopping Hamiltonian $H_{\rm hop}^{(\Lambda)}$ of the first term is slightly complicated to realize 
the reflection positivity \cite{JP} for the fermions. It is written in the form,  
\begin{equation}
\label{Hhop}
H_{\rm hop}^{(\Lambda)}=\sum_{i=1}^d H_{{\rm hop},i}^{(\Lambda)},
\end{equation}
with   
\begin{eqnarray}
\label{Hhop1}
H_{\rm hop,1}^{(\Lambda)}&=&i\kappa \sum_{\sigma=\uparrow,\downarrow}\sum_{\substack{x\in \Lambda \\ :\; x^{(1)}\ne L}} 
\left[(a_{x,\sigma})^\dagger a_{x+e_1,\sigma}-(a_{x+e_1,\sigma})^\dagger a_{x,\sigma}\right]\ret
&-&i\kappa \sum_{\sigma=\uparrow,\downarrow}\sum_{x^{(2)},\ldots,x^{(d)}}
\left[(a_{x_1^{+},\sigma})^\dagger a_{x_1^{-},\sigma}-(a_{x_1^{-},\sigma})^\dagger a_{x_1^{+},\sigma}\right]
\end{eqnarray}
for the hopping in the first direction and 
\begin{eqnarray}
\label{Hhop2}
H_{{\rm hop},i}^{(\Lambda)}&=&i\kappa \sum_{\sigma=\uparrow,\downarrow}\sum_{\substack{x\in\Lambda \\ :\; x^{(i)}\ne L}}
(-1)^{x^{(1)}+\cdots+x^{(i-1)}}\left[(a_{x,\sigma})^\dagger a_{x+e_i,\sigma}-(a_{x+e_i,\sigma})^\dagger a_{x,\sigma}\right]\ret
&-&i\kappa \sum_{\sigma=\uparrow,\downarrow}\sum_{x^{(1)},\ldots,x^{(i-1)},x^{(i+1)},\ldots,x^{(d)}}(-1)^{x^{(1)}+\cdots+x^{(i-1)}}
\left[(a_{x_i^{+},\sigma})^\dagger a_{x_i^{-},\sigma}-(a_{x_i^{-},\sigma})^\dagger a_{x_i^{+},\sigma}]\right]\ret
\end{eqnarray}
in the $i$-th direction for $i=2,\ldots,d$, 
where $e_i$ is the unit vector whose $i$-th component is $1$, $\kappa\in\re$ is the hopping amplitude, 
and the boundary sites $x_i^+,x_i^-$ of the lattice $\Lambda$ are given by 
$$
x_i^{+}=(x^{(1)},\ldots,x^{(i-1)},L,x^{(i+1)},\ldots,x^{(d)})
$$
and
$$ 
x_i^{-}=(x^{(1)},\ldots,x^{(i-1)},-L+1,x^{(i+1)},\ldots,x^{(d)}).
$$ 
We have imposed the anti-periodic boundary conditions for the hopping Hamiltonian $H_{\rm hop}^{(\Lambda)}$. 
In addition, the fluxes through unit squares \cite{LiebFlux} are all equal to $\pi$. 
The $\pi$ flux condition together with the anti-periodic boundary conditions satisfies 
the conditions for the canonical flux configuration in \cite{MN}. These are crucial to realize the reflection positivity 
for the present system. 

We write 
\begin{equation}
\label{TEV}
\langle\cdots\rangle_{\beta,B}^{(\Lambda)}:=\frac{1}{Z_{\beta,B}^{(\Lambda)}}{\rm Tr}\; (\cdots)e^{-\beta H^{(\Lambda)}(B)}
\end{equation}
for the thermal expectation value, where $Z_{\beta,B}^{(\Lambda)}:={\rm Tr}\; e^{-\beta H^{(\Lambda)}(B)}$ 
and $\beta\ge 0$ is the inverse temperature. The long-range order of the superconductivity is given by 
\begin{equation}
m_{\rm LRO}^{(\Lambda)}:=
\frac{1}{|\Lambda|}\sqrt{\langle [O^{(\Lambda)}]^2\rangle_{\beta,0}^{(\Lambda)}}
\end{equation}
with the external symmetry-breaking field $B=0$. We write 
\begin{equation}
m_{\rm LRO}:=\lim_{\Lambda\nearrow\ze^d}m_{\rm LRO}^{(\Lambda)}.
\end{equation}

\begin{theorem}
\label{existenceLRO}
Let $d\ge 3$, and set the external symmetry-breaking field to be $B=0$. 
Then, there exist small positive numbers $\hat{\kappa}$ and 
$\hat{g}'$, and  a large positive number $\beta_{\rm c}$ such that $m_{\rm LRO}>0$ for 
$|\kappa|/g\le \hat{\kappa}$ and $g'/g\le \hat{g}'$, and $\beta\ge \beta_{\rm c}$. 
Namely, the long-range order of the superconductivity exists on this parameter region with the strong coupling $g$. 
Here, $\hat{\kappa}$ and $\hat{g}'$ can be chosen to be independent of the model parameters, $\kappa, g, g'$, 
but $\beta_{\rm c}$ depends on the model parameters. 
\end{theorem} 

\noindent
{\it Remark:} When $\kappa=0$ and $g'=0$, the present model corresponds to the strong-coupling limit for the BCS-type model 
with the short-range interaction. The mean-field-type models were dealt in \cite{Thouless,MHK}. 
\medskip

The ground state in the finite volume is given by 
\begin{equation}
\omega_{B,g'}^{(\Lambda)}(\cdots):=\lim_{\beta\nearrow\infty}\langle\cdots\rangle_{\beta,B}^{(\Lambda)}. 
\end{equation}
We write 
\begin{equation}
\omega_{0,g'}(\cdots):={\rm weak}^\ast\mbox{-}\lim_{B\searrow 0}{\rm weak}^\ast\mbox{-}\lim_{\Lambda\nearrow\ze^d}
\omega_{B,g'}^{(\Lambda)}(\cdots). 
\end{equation}
For the parameters $\kappa, g, g'$ in the parameter region of Theorem~\ref{existenceLRO}, this is 
the symmetry-breaking infinite-volume ground state. Actually, there appears a spontaneous magnetization in the sense 
that \cite{KomaTasaki1} 
\begin{equation}
\frac{1}{|\Lambda|}\omega_{0,g'}(O^{(\Lambda)})\ge \mu >0
\end{equation}
for a large $\Lambda$, where $\mu$ is a positive constant. 
 
\begin{theorem}
Let $d\ge 3$. Suppose that the model parameters, $\kappa$ and $g'$, satisfy $|\kappa|/g\le\hat{\kappa}$ and 
$0<g'/g\le \hat{g}'$, where $\hat{\kappa}$ and $\hat{g}'$ are given in Theorem~\ref{existenceLRO}.  
Then, there exists a gapless quasi-local excitation above the infinite-volume ground state $\omega_{0,g'}(\cdots)$. 
Namely, a Nambu-Goldstone mode appears above the ground state.  
\end{theorem}

\noindent
Due to a technical reason, we need the assumption of the strict positivity of $g'$, i.e., $g'>0$.  
When $g'=0$, the existence of the Nambu-Goldstone mode is expected to hold as well. 

\Section{Local order parameters and interactions}

We will consider first the case with $g'=0$, i.e., without the Coulomb repulsion, for simplicity. 
As a preliminary step, we express the interaction Hamiltonian $H_{\rm int}^{(\Lambda)}$ in terms of the local order parameters 
in this section. In particular, an analogy \cite{BCS} to quantum antiferromagnets is very useful to apply the method of 
reflection positivity \cite{JP} to the present fermion system. 

We write 
\begin{equation}
\Gamma_x^{(+)}:=a_{x,\uparrow}^\dagger a_{x,\downarrow}^\dagger
\quad \mbox{and} \quad 
\Gamma_x^{(-)}:=a_{x,\downarrow}a_{x,\uparrow}.
\end{equation}
Clearly, one has $(\Gamma_x^{(-)})^\dagger=\Gamma_x^{(+)}$. In terms of these, we can write 
\begin{equation}
\label{HintGamma}
H_{\rm int}^{(\Lambda)}=g\sum_{\{x,y\}\subset\Lambda:|x-y|=1}[\Gamma_x^{(+)}\Gamma_y^{(-)}
+\Gamma_x^{(-)}\Gamma_y^{(+)}]
\end{equation}
for the interaction Hamiltonian (\ref{Hint}) with $g'=0$. Further, we introduce 
\begin{equation}
\Gamma_x^{(1)}:=\Gamma_x^{(+)}+\Gamma_x^{(-)}
\quad \mbox{and}\quad \Gamma_x^{(2)}:=i[\Gamma_x^{(+)}-\Gamma_x^{(-)}].
\end{equation}
Then, the order parameter $O^{(\Lambda)}$ of (\ref{O}) is written 
$$
O^{(\Lambda)}=\sum_{x\in\Lambda}(-1)^{x^{(1)}+\cdots+x^{(d)}}\Gamma_x^{(2)}.
$$

Since one has $\Gamma_x^{(1)}\Gamma_y^{(1)}+\Gamma_x^{(2)}\Gamma_y^{(2)}
=2[\Gamma_x^{(+)}\Gamma_y^{(-)}+\Gamma_x^{(-)}\Gamma_y^{(+)}]$, the interaction Hamiltonian of (\ref{HintGamma}) 
is written 
\begin{eqnarray}
\label{HintGamma12}
H_{\rm int}^{(\Lambda)}&=&\frac{g}{2}\sum_{\{x,y\}\subset\Lambda:|x-y|=1}
[\Gamma_x^{(1)}\Gamma_y^{(1)}+\Gamma_x^{(2)}\Gamma_y^{(2)}]\ret
&=&\frac{g}{4}\sum_{\{x,y\}\subset\Lambda:|x-y|=1}
\left\{[\Gamma_x^{(1)}+\Gamma_y^{(1)}]^2-[\Gamma_x^{(1)}]^2-[\Gamma_y^{(1)}]^2\right\}\ret
&-&\frac{g}{4}\sum_{\{x,y\}\subset\Lambda:|x-y|=1}
\left\{[\Gamma_x^{(2)}-\Gamma_y^{(2)}]^2-[\Gamma_x^{(2)}]^2-[\Gamma_y^{(2)}]^2\right\}\ret
&=&\frac{g}{4}\sum_{\{x,y\}\subset\Lambda:|x-y|=1}
[\Gamma_x^{(1)}+\Gamma_y^{(1)}]^2-\frac{g}{4}\sum_{\{x,y\}\subset\Lambda:|x-y|=1}[\Gamma_x^{(2)}-\Gamma_y^{(2)}]^2\ret
&-&\frac{dg}{2}\sum_{x\in\Lambda}\left\{[\Gamma_x^{(1)}]^2-[\Gamma_x^{(2)}]^2\right\}.
\end{eqnarray}

Let $\{h_m(x)\;|\;x\in\Lambda,m=1,2,\ldots,d\}$ be $d$ real-valued functions on the lattice $\Lambda$. 
We define 
\begin{eqnarray}
\label{Hinth}
H_{\rm int}^{(\Lambda)}(h)&:=&\frac{g}{4}\sum_{x\in\Lambda}\sum_{m=1}^d
[\Gamma_x^{(1)}+\Gamma_{x+e_m}^{(1)}+(-1)^{x^{(1)}+\cdots+x^{(d)}}h_m(x)]^2\ret
&-&\frac{g}{4}\sum_{x\in\Lambda}\sum_{m=1}^d[\Gamma_x^{(2)}-\Gamma_{x+e_m}^{(2)}]^2
-\frac{dg}{2}\sum_{x\in\Lambda}\left\{[\Gamma_x^{(1)}]^2-[\Gamma_x^{(2)}]^2\right\},
\end{eqnarray}
and 
\begin{equation}
H^{(\Lambda)}(B,h):=H_{\rm hop}^{(\Lambda)}+H_{\rm int}^{(\Lambda)}(h)-BO^{(\Lambda)}, 
\end{equation}
where $e_m$ is the unit vector whose $m$-th component is $1$. Clearly, $H_{\rm int}^{(\Lambda)}(0)=H_{\rm int}^{(\Lambda)}$ 
and $H^{(\Lambda)}(B,0)=H^{(\Lambda)}(B)$. 

We will also use the following relations:  
\begin{eqnarray}
\label{GammaSq}
[\Gamma_x^{(1)}]^2&=&(a_{x,\uparrow}^\dagger a_{x,\downarrow}^\dagger+a_{x,\downarrow}a_{x,\uparrow})
(a_{x,\uparrow}^\dagger a_{x,\downarrow}^\dagger+a_{x,\downarrow}a_{x,\uparrow})\ret
&=&a_{x,\uparrow}^\dagger a_{x,\downarrow}^\dagger a_{x,\downarrow}a_{x,\uparrow}
+a_{x,\downarrow}a_{x,\uparrow}a_{x,\uparrow}^\dagger a_{x,\downarrow}^\dagger\ret
&=&n_{x,\uparrow}n_{x,\downarrow}+(1-n_{x,\downarrow})(1-n_{x,\uparrow})\ret
&=&2n_{x,\uparrow}n_{x,\downarrow}-n_{x,\uparrow}-n_{x,\downarrow}+1
\end{eqnarray}
and 
\begin{eqnarray}
[\Gamma_x^{(1)},\Gamma_x^{(2)}]&=&i[\Gamma_x^{(+)}+\Gamma_x^{(-)},\Gamma_x^{(+)}-\Gamma_x^{(-)}]\ret
&=&-i[\Gamma_x^{(+)},\Gamma_x^{(-)}]+i[\Gamma_x^{(-)},\Gamma_x^{(+)}]\ret
&=&2i[\Gamma_x^{(-)},\Gamma_x^{(+)}]\ret
&=&2i(a_{x,\downarrow}a_{x,\uparrow}a_{x,\uparrow}^\dagger a_{x,\downarrow}^\dagger 
-a_{x,\uparrow}^\dagger a_{x,\downarrow}^\dagger a_{x,\downarrow}a_{x,\uparrow})\ret
&=&2i[(1-n_{x,\downarrow})(1-n_{x,\uparrow})-n_{x,\uparrow}n_{x,\downarrow}]=2i\Gamma_x^{(3)},
\end{eqnarray}
where we have written  
\begin{equation}
\Gamma_x^{(3)}:=1-n_{x,\uparrow}-n_{x,\downarrow}.
\end{equation}
Similarly, one has 
\begin{equation}
[\Gamma_x^{(2)},\Gamma_x^{(3)}]=2i\Gamma_x^{(1)}
\quad \mbox{and}\quad [\Gamma_x^{(3)},\Gamma_x^{(1)}]=2i\Gamma_x^{(2)}.
\end{equation}
Clearly, these commutation relations are the same as those between the spin-$1/2$ operators.  
The relation between $\Gamma_x^{(1)}$ and $\Gamma_x^{(2)}$ is given by 
\begin{equation}
e^{-i\pi \Gamma_x^{(3)}/4}\Gamma_x^{(1)}e^{i\pi\Gamma_x^{(3)}/4}=\Gamma_x^{(2)}
\quad \mbox{or}\quad 
e^{-i\pi \Gamma_x^{(3)}/4}\Gamma_x^{(2)}e^{i\pi\Gamma_x^{(3)}/4}=-\Gamma_x^{(1)}.
\end{equation}
The derivation is given as (\ref{U1rotation}) in Appendix~\ref{appendix:U(1)rotation}. 
Namely, the operator $\Gamma_x^{(3)}$ is the generator of the U(1) rotation. 
Therefore, the operator $\sum_{x\in\Lambda}(-1)^{x^{(1)}+\cdots+x^{(d)}}\Gamma_x^{(1)}$ is also an order parameter.

\Section{Gauge transformations}
\label{GaugeTrans}

The present system does not have usual translational invariance because of the anti-periodic boundary conditions 
for the hopping terms $H_{\rm hop}^{(\Lambda)}$ in the Hamiltonian $H^{(\Lambda)}(B)$. 
In addition, the hopping amplitudes are different between the spatial directions. 
However, we can change the locations of the bonds having the opposite sign of the hopping amplitudes 
by using a gauge transformation. We can also interchange the hopping amplitudes between the spatial directions 
by a gauge transformation. 
Thus, in the sense of the gauge equivalence, the present system has translational invariance and 
direction-independence of the hopping amplitudes.  
In this section, we show these two properties. 

\subsection{Gauge transformations for the boundary condition}
\label{GaugeTransBC}

Consider a unitary operator, 
$$
U_{{\rm BC},1}^{(\Lambda)}(L\rightarrow \ell)
:=\prod_{\sigma=\uparrow,\downarrow}\prod_{\substack{x\in\Lambda \\:\; \ell\le x^{(1)}\le L}}
e^{i\pi n_{x\sigma}},
$$
which yields the transformation,  
$$
(U_{{\rm BC},1}^{(\Lambda)}(L\rightarrow \ell))^\dagger a_{x,\sigma}U_{\rm BC}^{(\Lambda)}(L\rightarrow \ell)
=-a_{x,\sigma},
$$
for $x$ satisfying $\ell\le x^{(1)}\le L$. Therefore, this transformation $U_{{\rm BC},1}^{(\Lambda)}(L\rightarrow \ell)$ changes 
the locations of the bonds having the opposite sign of the hopping amplitudes from the bonds $\{x_1^-,x_1^+\}$ to 
the bonds $\{(\ell-1,x^{(2)},\ldots,x^{(d)}),(\ell,x^{(2)},\ldots,x^{(d)})\}$ along the $x^{(1)}$-direction. 
Clearly, the rest of the terms in the Hamiltonian do not change under this transformation, 
and similar transformations are also possible in all the other directions. 
We write $U_{{\rm BC},i}^{(\Lambda)}(L\rightarrow \ell)$ for the unitary transformation in the $x^{(i)}$-direction, 
$i=1,2,\ldots,d$.  

We write $\mathcal{T}_m$ for the lattice shift transformation by two lattice units which is defined by 
$\mathcal{T}_m(a_{x,\sigma})=a_{x+2e_m,\sigma}$ and $\mathcal{T}_m(a_{x,\sigma}^\dagger)=a_{x+2e_m,\sigma}^\dagger$ 
for $m=1,2,\ldots,d$. Consider 
\begin{equation}
\langle \Gamma_x^{(j)}\Gamma_y^{(j)}\rangle_{\beta,B}^{(\Lambda)} \quad \mbox{for \ } j=1,2, 
\end{equation}
where $\langle \cdots\rangle_{\beta,B}^{(\Lambda)}$ is given by (\ref{TEV}). One has 
\begin{eqnarray}
\langle \Gamma_x^{(j)}\Gamma_y^{(j)}\rangle_{\beta,B}^{(\Lambda)}
&=&{\rm Tr}\; \mathcal{T}_m(\Gamma_x^{(j)}\Gamma_y^{(j)})e^{-\beta \mathcal{T}_m(H^{(\Lambda)}(B))}/Z_{\beta,B}^{(\Lambda)}\ret
&=&{\rm Tr}\; \Gamma_{x+2e_m}^{(j)}\Gamma_{y+2e_m}^{(j)}e^{-\beta \mathcal{T}_m(H^{(\Lambda)}(B))}/Z_{\beta,B}^{(\Lambda)}.
\end{eqnarray}
Clearly, the boundary condition of the Hamiltonian $\mathcal{T}_m(H^{(\Lambda)}(B))$ in the right-hand side 
is different from that of the Hamiltonian $H^{(\Lambda)}(B)$. But, we can change the boundary condition by using 
the unitary transformation $U_{{\rm BC},m}^{(\Lambda)}(L\rightarrow \ell)$, which does not change $\Gamma_x^{(j)}$.  
Therefore, we have 
\begin{equation}
\label{transinvGamma}
\langle \Gamma_x^{(j)}\rangle_{\beta,B}^{(\Lambda)}=\langle \Gamma_{x+2e_m}^{(j)}\rangle_{\beta,B}^{(\Lambda)}
\end{equation}
and 
\begin{equation}
\label{transinvGammaGamma}
\langle \Gamma_x^{(j)}\Gamma_y^{(j)}\rangle_{\beta,B}^{(\Lambda)}
=\langle \Gamma_{x+2e_m}^{(j)}\Gamma_{y+2e_m}^{(j)}\rangle_{\beta,B}^{(\Lambda)}
\end{equation}
for $m=1,2,\ldots,d$ and $j=1,2$. We will use this lattice shift invariance later.

\subsection{Gauge transformations for the hoping amplitudes} 
\label{GaugeTransHA}

The hopping amplitudes in the hopping Hamiltonian $H_{\rm hop}^{(\Lambda)}$ are asymmetric in all the directions. 
But we can also interchange the amplitudes as follows: 
We first introduce a unitary transformation, 
\begin{equation}
\label{defU3}
U_{\rm HA}^{(\Lambda)}(i,j):=\prod_{\sigma=\uparrow,\downarrow}
\prod_{\substack{x\in\Lambda \\ :\; x^{(i)}={\rm odd} \ {\rm and} \ x^{(j)}={\rm odd}}}
e^{i\pi n_{x,\sigma}},\quad \mbox{for \ } i\ne j. 
\end{equation}
This transformation changes only two hopping amplitudes in the $x^{(i)}$ and $x^{(j)}$ directions. 
The two hopping terms at each bond are written 
\begin{equation}
\label{ithhopamp}
(-1)^{x^{(1)}+\cdots+x^{(i-1)}}(a_{x,\sigma}^\dagger a_{x+e_i,\sigma}-a_{x+e_i,\sigma}^\dagger a_{x,\sigma})
\end{equation}
and 
\begin{equation}
\label{jthhopamp}
(-1)^{x^{(1)}+\cdots+x^{(j-1)}}(a_{x,\sigma}^\dagger a_{x+e_j,\sigma}-a_{x+e_j,\sigma}^\dagger a_{x,\sigma})
\end{equation}
except for the constant factor $\pm i\kappa$. 

Consider first the former (\ref{ithhopamp}). When $x^{(j)}={\rm even}$,   
it does not change under the transformation $U_{\rm HA}^{(\Lambda)}(i,j)$ by definition. 
On the other hand, when $x^{(j)}={\rm odd}$, one of $x^{(i)}$ and $x^{(i)}+e_i^{(i)}$ is odd, and 
the other is even. Therefore, one has 
\begin{eqnarray}
\label{HACi}
& &(U_{\rm HA}^{(\Lambda)}(i,j))^\dagger 
(-1)^{x^{(1)}+\cdots+x^{(i-1)}}(a_{x,\sigma}^\dagger a_{x+e_i,\sigma}-a_{x+e_i,\sigma}^\dagger a_{x,\sigma})
U_{\rm HA}^{(\Lambda)}(i,j)\ret
&=&
(-1)^{x^{(1)}+\cdots+x^{(i-1)}}(-1)^{x^{(j)}}(a_{x,\sigma}^\dagger a_{x+e_i,\sigma}-a_{x+e_i,\sigma}^\dagger a_{x,\sigma}).
\end{eqnarray}

Next, consider the latter (\ref{jthhopamp}). When $x^{(i)}={\rm even}$, it does not change 
under the transformation $U_{\rm HA}^{(\Lambda)}(i,j)$ by definition as well.
On the other hand, when $x^{(i)}={\rm odd}$, one of $x^{(j)}$ and $x^{(j)}+e_j^{(j)}$ is odd, and 
the other is even. Therefore, one has 
\begin{eqnarray}
\label{HACj}
& &(U_{\rm HA}^{(\Lambda)}(i,j))^\dagger 
(-1)^{x^{(1)}+\cdots+x^{(j-1)}}(a_{x,\sigma}^\dagger a_{x+e_j,\sigma}-a_{x+e_j,\sigma}^\dagger a_{x,\sigma})
U_{\rm HA}^{(\Lambda)}(i,j)\ret
&=&
(-1)^{x^{(1)}+\cdots+x^{(j-1)}}(-1)^{x^{(i)}}(a_{x,\sigma}^\dagger a_{x+e_j,\sigma}-a_{x+e_j,\sigma}^\dagger a_{x,\sigma}). 
\end{eqnarray}

By using these two relations, we can interchange the roles of the hopping amplitudes 
in the $x^{(1)}$ and $x^{(j)}$ directions as follows: We define 
\begin{equation}
U_{\rm HA}^{(\Lambda)}(j\rightarrow 1):=U_{\rm HA}^{(\Lambda)}(j,j-1)U_{\rm HA}^{(\Lambda)}(j,j-1)U_{\rm HA}^{(\Lambda)}(j,j-2)
\cdots U_{\rm HA}^{(\Lambda)}(j,2)U_{\rm HA}^{(\Lambda)}(j,1).
\end{equation}
Then, we have 
\begin{eqnarray}
& &(U_{\rm HA}^{(\Lambda)}(j\rightarrow 1))^\dagger 
(-1)^{x^{(1)}+\cdots+x^{(i-1)}}(a_{x,\sigma}^\dagger a_{x+e_i,\sigma}-a_{x+e_i,\sigma}^\dagger a_{x,\sigma})
U_{\rm HA}^{(\Lambda)}(j\rightarrow 1)\ret
&=&
(-1)^{x^{(1)}+\cdots+x^{(i-1)}}(-1)^{x^{(j)}}(a_{x,\sigma}^\dagger a_{x+e_i,\sigma}-a_{x+e_i,\sigma}^\dagger a_{x,\sigma})
\end{eqnarray}
for $i<j$ from the above (\ref{HACi}), and  
\begin{eqnarray}
& &(U_{\rm HA}^{(\Lambda)}(j\rightarrow 1))^\dagger 
(-1)^{x^{(1)}+\cdots+x^{(j-1)}}(a_{x,\sigma}^\dagger a_{x+e_j,\sigma}-a_{x+e_j,\sigma}^\dagger a_{x,\sigma})
U_{\rm HA}^{(\Lambda)}(j\rightarrow 1)\ret
&=&(a_{x,\sigma}^\dagger a_{x+e_j,\sigma}-a_{x+e_j,\sigma}^\dagger a_{x,\sigma}) 
\end{eqnarray}
{from} (\ref{HACj}). Since the rest of the terms in the Hamiltonian $H^{(\Lambda)}$ do not change under 
this transformation $U_{\rm HA}^{(\Lambda)}(j\rightarrow 1)$, these are the desired results. 
Thus, in the sense of this gauge equivalence, the hopping amplitudes are equivalent in all the directions. 

In other words, the transformation $U_{\rm HA}^{(\Lambda)}(j\rightarrow 1)$ induces the permutation of the coordinate, 
$$
\begin{pmatrix}
x^{(1)} \\
x^{(2)} \\
\vdots \\
x^{(j-1)} \\
x^{(j)} \\
\end{pmatrix}
\rightarrow 
\begin{pmatrix}
x^{(j)} \\
x^{(1)} \\
x^{(2)} \\
\vdots \\
x^{(j-1)} \\
\end{pmatrix}.
$$
We write $\mathcal{P}_{j\rightarrow 1}$ for the transformation of this permutation. 
Then, one has 
\begin{eqnarray*}
\sum_{x\in \Lambda}\langle \Gamma_x^{(1)}\Gamma_{x+e_1}^{(1)}\rangle_{\beta,B}^{(\Lambda)}
&=&\sum_{x\in\Lambda}{\rm Tr}\; \mathcal{P}_{j\rightarrow 1}(\Gamma_x^{(1)}\Gamma_{x+e_1}^{(1)})
e^{-\beta\mathcal{P}_{j\rightarrow 1}(H^{(\Lambda)})}/Z_{\beta,B}^{(\Lambda)}\ret
&=&\sum_{x\in\Lambda}{\rm Tr}\; \Gamma_x^{(1)}\Gamma_{x+e_j}^{(1)}
e^{-\beta\mathcal{P}_{j\rightarrow 1}(H^{(\Lambda)})}/Z_{\beta,B}^{(\Lambda)}.
\end{eqnarray*}
Since the inverse of $U_{\rm HA}^{(\Lambda)}(j\rightarrow 1)$ does not change $\Gamma_x^{(1)}$, one obtains 
\begin{equation}
\label{direcIndep}
\sum_{x\in \Lambda}\langle \Gamma_x^{(1)}\Gamma_{x+e_1}^{(1)}\rangle_{\beta,B}^{(\Lambda)}
=\sum_{x\in \Lambda}\langle \Gamma_x^{(1)}\Gamma_{x+e_j}^{(1)}\rangle_{\beta,B}^{(\Lambda)}
\end{equation}
for any $j$. We will use this relation later. 

\Section{Reflection positivity -- Real space}
\label{Sec:RPRS}

Our goal in this section is to prove the Gaussian domination bound (\ref{GaussDomi}) 
in Theorem~\ref{theorem:GaussDomi} below 
by using the method of the reflection positivity \cite{JP} for the present system. 

Since the sidelengths of the cube $\Lambda$ of (\ref{Lambda}) are all the even integer $2L$, 
the lattice $\Lambda$ is invariant under a reflection $\vartheta$ in a plane $\Pi$ 
which is normal to a coordinate direction and intersects no sites in $\Lambda$, i.e., $\vartheta(\Lambda)=\Lambda$. 
Clearly, the lattice $\Lambda$ can be decomposed into two parts $\Lambda=\Lambda_-\cup\Lambda_+$, 
where $\Lambda_\pm$ denote the set of the sites on the $\pm$ side of the plane $\Pi$. 
The reflection $\vartheta$ maps $\Lambda_\pm$ into $\Lambda_\mp$, i.e., $\vartheta(\Lambda_\pm)=\Lambda_\mp$. 
Let $\Omega$ be a subset of $\Lambda$, i.e., $\Omega\subset\Lambda$. 
We write $\mathfrak{A}(\Omega)$ for the algebra generated by $a_{x,\sigma}$ and $a_{x',\sigma'}^\dagger$ for $x,x'\in\Omega$, 
$\sigma,\sigma'=\uparrow,\downarrow$. 
We also write $\mathfrak{A}=\mathfrak{A}(\Lambda)$, and $\mathfrak{A}_\pm=\mathfrak{A}(\Lambda_\pm)$. 

Following \cite{JP}, we consider an anti-linear representation of the reflection $\vartheta$ on the fermion 
Hilbert space, which is also denoted by $\vartheta$. 
The anti-linear map, $\vartheta:\mathfrak{A}_\pm\rightarrow\mathfrak{A}_\mp$, is defined by 
\begin{equation}
\label{defvartheta}
\vartheta(a_{x,\sigma})=a_{\vartheta(x),\sigma}\quad \mbox{and} \quad 
\vartheta(a_{x,\sigma}^\dagger)=a_{\vartheta(x),\sigma}^\dagger. 
\end{equation}
Here, we have used the fact that the fermion operators $a_{x,\sigma}$ have a real representation. 
For $\mathcal{A},\mathcal{B}\in\mathfrak{A}$, 
\begin{equation}
\vartheta(\mathcal{A}\mathcal{B})=\vartheta(\mathcal{A})\vartheta(\mathcal{B})
\quad \mbox{and}\quad \vartheta(\mathcal{A})^\dagger=\vartheta(\mathcal{A}^\dagger). 
\end{equation}

Let us introduce Majorana fermion operators,
\begin{equation}
\label{xi}
\xi_{x,\sigma}:=a_{x,\sigma}^\dagger + a_{x,\sigma}
\end{equation}
and 
\begin{equation}
\label{eta}
\eta_{x,\sigma}:=i(a_{x,\sigma}^\dagger - a_{x,\sigma}),
\end{equation}
for $x,x'\in\Lambda$, $\sigma,\sigma'=\uparrow,\downarrow$. 
These satisfy the anticommutation relations, 
\begin{equation}
\{\xi_{x,\sigma},\xi_{x',\sigma'}\}=2\delta_{x,x'}\delta_{\sigma,\sigma'},
\end{equation}
\begin{equation}
\{\eta_{x,\sigma},\eta_{x',\sigma'}\}=2\delta_{x,x'}\delta_{\sigma,\sigma'},
\end{equation}
and 
\begin{equation}
\{\xi_{x,\sigma},\eta_{x',\sigma'}\}=0.
\end{equation}
By the definition (\ref{defvartheta}) of the anti-linear map $\vartheta$ of the reflection, one has 
\begin{equation}
\vartheta(\xi_{x,\sigma})=\xi_{\vartheta(x),\sigma}
\quad \mbox{and}\quad \vartheta(\eta_{x,\sigma})=-\eta_{\vartheta(x),\sigma},
\end{equation}
and 
\begin{equation}
\xi_{x,\sigma}^\dagger=\xi_{x,\sigma},\quad \mbox{and}\quad \eta_{x,\sigma}^\dagger=\eta_{x,\sigma}.
\end{equation}
The fermion operators are written as 
\begin{equation}
\label{MajoranaRep}
a_{x,\sigma}=\frac{1}{2}(\xi_{x,\sigma}+i\eta_{x,\sigma}), 
\quad a_{x,\sigma}^\dagger=\frac{1}{2}(\xi_{x,\sigma}-i\eta_{x,\sigma}).
\end{equation}

Since we have imposed the periodic boundary conditions for the lattice $\Lambda$, 
there are two planes which divide the lattice $\Lambda$ 
into the two halves, $\Lambda_-$ and $\Lambda_+$. We denote the two planes collectively by $\tilde{\Pi}$. 
Consider first the reflection in the plane $\Pi$ normal to the $x^{(1)}$ direction. 
Since we can change the locations of the bonds with the opposite sign of the hopping amplitudes 
due to the anti-periodic boundary conditions by using a gauge transformation as shown in Sec.~\ref{GaugeTransBC}, 
we can take 
\begin{equation}
\Lambda_-=\{x|-L+1\le x^{(1)}\le 0\}\quad \mbox{and}\quad \Lambda_+=\{x|\; 1\le x^{(1)}\le L\}
\end{equation}
without loss of generality. 
Then, the bonds crossing the plane $\tilde{\Pi}$ are given by 
$$
\{(0,x^{(2)},\ldots,x^{(d)}),(1,x^{(2)},\ldots,x^{(d)})\}
$$ 
and 
$$
\{(-L+1,x^{(2)},\ldots,x^{(d)}),(L,x^{(2)},\ldots,x^{(d)})\}
$$ 
for $x^{(2)},\ldots,x^{(d)}\in\{-L+1,-L+2,\ldots,-1,0,1,\ldots,L-1,L\}$.

We introduce some unitary transformations as follows: 
\begin{equation}
\label{defU1j}
U_{1,j}^{(\Lambda)}
:=\prod_{\substack{x\in\Lambda,\; \sigma=\uparrow,\downarrow \\ :\;x^{(j)}={\rm even}}}e^{(i\pi/2)n_{x,\sigma}}
\quad \mbox{for \ } j=2,3,\ldots,d, 
\end{equation}
and 
\begin{equation}
\label{defU1}
U_1^{(\Lambda)}:=\prod_{j=2}^d U_{1,j}^{(\Lambda)}. 
\end{equation}
Since one has 
\begin{equation}
e^{-(i\pi/2)n_{x,\sigma}}a_{x,\sigma}e^{(i\pi/2)n_{x,\sigma}}=ia_{x,\sigma},
\end{equation}
we have 
\begin{equation}
(U_{1,j}^{(\Lambda)})^\dagger a_{x,\sigma}U_{1,j}^{(\Lambda)}=
\begin{cases}
ia_{x,\sigma}, &  \mbox{for}\ x^{(j)}={\rm even}; \\
a_{x,\sigma}, & \mbox{for}\ x^{(j)}={\rm odd}. 
\end{cases}
\end{equation}

Next, we introduce \cite{FILS} 
\begin{equation}
\label{uxsigma}
u_{x,\sigma}:=\left[\prod_{\substack{y\in\Lambda,\; \sigma'=\uparrow,\downarrow \\ :\; y\ne x\; {\rm or}\; \sigma'\ne \sigma}}
(-1)^{n_{y,\sigma'}}\right](a_{x,\sigma}^\dagger+a_{x,\sigma}). 
\end{equation}
Then, 
\begin{equation}
(u_{x,\sigma})^\dagger a_{y,\sigma'}u_{x,\sigma}=
\begin{cases}
a_{x,\sigma}^\dagger, & \mbox{for}\; y=x \ \mbox{and} \ \sigma'=\sigma;\\
a_{y,\sigma'}, & \mbox{otherwise}.
\end{cases}
\end{equation}
By using these operators, we further introduce \cite{FILS} 
\begin{equation}
\label{defU2}
U_{\rm odd}^{(\Lambda)}:=\prod_{\sigma=\uparrow,\downarrow}\prod_{x\in\Lambda_{\rm odd}}u_{x,\sigma},
\end{equation}
where 
\begin{equation}
\Lambda_{\rm odd}:=\{x |\; x^{(1)}+x^{(2)}+\cdots +x^{(d)}={\rm odd}\}.
\end{equation}
Immediately, 
\begin{equation}
\label{U2}
(U_{\rm odd}^{(\Lambda)})^\dagger a_{x,\sigma}U_{\rm odd}^{(\Lambda)}=
\begin{cases}
a_{x,\sigma}^\dagger, & \mbox{for}\; x\in \Lambda_{\rm odd};\\
a_{x,\sigma}, & \mbox{otherwise}.  
\end{cases}
\end{equation}
We write 
\begin{equation}
\label{U}
\tilde{U}_1^{(\Lambda)}:=U_1^{(\Lambda)}U_{\rm odd}^{(\Lambda)}.
\end{equation}

Note that 
\begin{equation}
(U_{1,j}^{(\Lambda)})^\dagger H_{{\rm hop},i}^{(\Lambda)}U_{1,j}^{(\Lambda)}
=H_{{\rm hop},i}^{(\Lambda)}
\end{equation}
for $i\ne j$ because the term, 
$(a_{x,\sigma}^\dagger a_{x+e_i,\sigma}-a_{x+e_i,\sigma}^\dagger a_{x,\sigma})$, does not change under 
the transformation from $x^{(j)}=(x+e_i)^{(j)}=x^{(j)}+e_i^{(j)}=x^{(j)}$. 
When $i=j$, we have 
\begin{eqnarray}
& &(U_{1,j}^{(\Lambda)})^\dagger H_{{\rm hop},j}^{(\Lambda)}U_{1,j}^{(\Lambda)}\ret
&=&\kappa \sum_{\sigma=\uparrow,\downarrow}\sum_{\substack{x\in\Lambda \\ :\; x^{(j)}\ne L}}
(-1)^{x^{(1)}+\cdots +x^{(j)}}(a_{x,\sigma}^\dagger a_{x+e_j,\sigma}+a_{x+e_j,\sigma}^\dagger a_{x,\sigma})\ret
&-&\kappa \sum_{\sigma=\uparrow,\downarrow}\sum_{x^{(1)},\ldots,x^{(j-1)},x^{(j+1)},\ldots,x^{(d)}}
(-1)^{x^{(1)}+\cdots +x^{(j-1)}+L}
(a_{x_j^+,\sigma}^\dagger a_{x_j^-,\sigma}+a_{x_j^-,\sigma}^\dagger a_{x_j^+,\sigma})\ret 
\end{eqnarray}
{from} the expression (\ref{Hhop2}) and the definitions, $x_j^{+}=(x^{(1)},\ldots,x^{(j-1)},L,x^{(j+1)},\ldots,x^{(d)})$ 
and $x_j^-=(x^{(1)},\ldots,x^{(j-1)},-L+1,x^{(j+1)},\ldots,x^{(d)})$. 
{From} these observations, one has 
\begin{equation}
\label{Hhop1invU1}
(U_1^{(\Lambda)})^\dagger H_{{\rm hop},1}^{(\Lambda)}U_1^{(\Lambda)}=H_{{\rm hop},1}^{(\Lambda)}
\end{equation}
and 
\begin{eqnarray}
& &(U_{1}^{(\Lambda)})^\dagger H_{{\rm hop},j}^{(\Lambda)}U_{1}^{(\Lambda)}\ret
&=&\kappa \sum_{\sigma=\uparrow,\downarrow}\sum_{\substack{x\in\Lambda \\ :\; x^{(j)}\ne L}}
(-1)^{x^{(1)}+\cdots +x^{(j)}}(a_{x,\sigma}^\dagger a_{x+e_j,\sigma}+a_{x+e_j,\sigma}^\dagger a_{x,\sigma})\ret
&-&\kappa \sum_{\sigma=\uparrow,\downarrow}\sum_{x^{(1)},\ldots,x^{(j-1)},x^{(j+1)},\ldots,x^{(d)}}
(-1)^{x^{(1)}+\cdots +x^{(j-1)}+L}
(a_{x_j^+,\sigma}^\dagger a_{x_j^-,\sigma}+a_{x_j^-,\sigma}^\dagger a_{x_j^+,\sigma})\ret 
\end{eqnarray}
for $j=2,\ldots,d$.  
Further, by using (\ref{defU1}), (\ref{U2}) and (\ref{U}), we have 
\begin{eqnarray}
\tilde{H}_{{\rm hop},j}^{(\Lambda)}&:=&(\tilde{U}_1^{(\Lambda)})^\dagger H_{{\rm hop},j}^{(\Lambda)}\tilde{U}_1^{(\Lambda)}\ret
&=&\kappa \sum_{\sigma=\uparrow,\downarrow}\sum_{\substack{x\in\Lambda \\ :\; x^{(j)}\ne L}}
(-1)^{x^{(j+1)}+\cdots +x^{(d)}}(a_{x,\sigma}^\dagger a_{x+e_j,\sigma}^\dagger+a_{x+e_j,\sigma} a_{x,\sigma})\ret
&-&\kappa \sum_{\sigma=\uparrow,\downarrow}\sum_{x^{(1)},\ldots,x^{(j-1)},x^{(j+1)},\ldots,x^{(d)}}
(-1)^{x^{(j+1)}+\cdots +x^{(d)}}
(a_{x_j^+,\sigma}^\dagger a_{x_j^-,\sigma}^\dagger+a_{x_j^-,\sigma} a_{x_j^+,\sigma})\ret
\end{eqnarray}
for $j=2,3,\ldots,d-1$, and  
\begin{eqnarray}
\tilde{H}_{{\rm hop},d}^{(\Lambda)}&:=&(\tilde{U}_1^{(\Lambda)})^\dagger H_{{\rm hop},d}^{(\Lambda)}\tilde{U}_1^{(\Lambda)}\ret
&=&\kappa \sum_{\sigma=\uparrow,\downarrow}\sum_{\substack{x\in\Lambda \\ :\; x^{(d)}\ne L}}
(a_{x,\sigma}^\dagger a_{x+e_d,\sigma}^\dagger+a_{x+e_d,\sigma} a_{x,\sigma})\ret
&-&\kappa \sum_{\sigma=\uparrow,\downarrow}\sum_{x^{(1)},\ldots,x^{(d-1)}}
(a_{x_d^+,\sigma}^\dagger a_{x_d^-,\sigma}^\dagger+a_{x_d^-,\sigma} a_{x_d^+,\sigma}).\ret
\end{eqnarray}
These Hamiltonians can be decomposed into two parts, 
\begin{equation}
\label{tildeHhop2decomp}
\tilde{H}_{{\rm hop},j}^{(\Lambda)}=\tilde{H}_{{\rm hop},j,-}^{(\Lambda)}+\tilde{H}_{{\rm hop},j,+}^{(\Lambda)}
\quad \mbox{for \ } j=2,\ldots,d, 
\end{equation}
where the two terms satisfy $\tilde{H}_{{\rm hop},j,\pm}^{(\Lambda)}\in \mathfrak{A}_\pm$ and 
$\vartheta(\tilde{H}_{{\rm hop},j,-}^{(\Lambda)})=\tilde{H}_{{\rm hop},j,+}^{(\Lambda)}$. 

Next, consider the hopping Hamiltonian $H_{\rm hop,1}^{(\Lambda)}$ of (\ref{Hhop1}). 
By using the Majorana representation (\ref{MajoranaRep}), each hopping can be written  
\begin{equation}
\label{hoppingMajoranaRep}
i\kappa(a_{x,\sigma}^\dagger a_{y,\sigma}-a_{y,\sigma}^\dagger a_{x,\sigma})
=\frac{i\kappa}{2}(\xi_{x,\sigma}\xi_{y,\sigma}+\eta_{x,\sigma}\eta_{y,\sigma}).
\end{equation}
{From} (\ref{xi}), (\ref{eta}) and (\ref{U2}), one has 
\begin{equation}
(U_{\rm odd}^{(\Lambda)})^\dagger \xi_{x,\sigma}U_{\rm odd}^{(\Lambda)}=\xi_{x,\sigma}, 
\end{equation}
and 
\begin{equation}
(U_{\rm odd}^{(\Lambda)})^\dagger \eta_{x,\sigma}U_{\rm odd}^{(\Lambda)}=
\begin{cases}
-\eta_{x,\sigma}, & \mbox{for}\ x\in\Lambda_{\rm odd};\\
\eta_{x,\sigma},  & \mbox{otherwise}.
\end{cases}
\end{equation}
Combining these with (\ref{hoppingMajoranaRep}), we obtain 
\begin{equation}
(U_{\rm odd}^{(\Lambda)})^\dagger i\kappa(a_{x,\sigma}^\dagger a_{y,\sigma}-a_{y,\sigma}^\dagger a_{x,\sigma})U_{\rm odd}^{(\Lambda)}
=\frac{i\kappa}{2}(\xi_{x,\sigma}\xi_{y,\sigma}-\eta_{x,\sigma}\eta_{y,\sigma})
\end{equation}
for $x,y$ satisfying $|x-y|=1$. This yields  
\begin{eqnarray}
\label{tildeHhop1}
\tilde{H}_{\rm hop,1}^{(\Lambda)}:=(\tilde{U}_1^{(\Lambda)})^\dagger H_{\rm hop,1}^{(\Lambda)}\tilde{U}_1^{(\Lambda)}
&=&\frac{i\kappa}{2} \sum_{\sigma=\uparrow,\downarrow}\sum_{\substack{x\in \Lambda \\ :\; x^{(1)}\ne L}} 
(\xi_{x,\sigma}\xi_{x+e_1,\sigma}-\eta_{x,\sigma}\eta_{x+e_1,\sigma})\ret
&-&\frac{i\kappa}{2} \sum_{\sigma=\uparrow,\downarrow}\sum_{x^{(2)},\ldots,x^{(d)}}
(\xi_{x_1^+,\sigma}\xi_{x_1^-,\sigma}-\eta_{x_1^+,\sigma}\eta_{x_1^-,\sigma}),
\end{eqnarray}
where we have used (\ref{U}) and (\ref{Hhop1invU1}). 
This Hamiltonian can be decomposed into three parts, 
\begin{equation}
\label{tildeHhop1decomp}
\tilde{H}_{\rm hop,1}^{(\Lambda)}=\tilde{H}_{\rm hop,1,-}^{(\Lambda)}+\tilde{H}_{\rm hop,1,0}^{(\Lambda)}
+\tilde{H}_{\rm hop,1,+}^{(\Lambda)},
\end{equation}
where 
\begin{equation}
\tilde{H}_{\rm hop,1,-}^{(\Lambda)}
:=\frac{i\kappa}{2} \sum_{\sigma=\uparrow,\downarrow}\sum_{\substack{x\in \Lambda_- \\ :\; x+e_1\in\Lambda_-}} 
(\xi_{x,\sigma}\xi_{x+e_1,\sigma}-\eta_{x,\sigma}\eta_{x+e_1,\sigma}),
\end{equation}
\begin{equation}
\tilde{H}_{\rm hop,1,+}^{(\Lambda)}
:=\frac{i\kappa}{2} \sum_{\sigma=\uparrow,\downarrow}\sum_{\substack{x\in \Lambda_+ \\ :\; x+e_1\in\Lambda_+}} 
(\xi_{x,\sigma}\xi_{x+e_1,\sigma}-\eta_{x,\sigma}\eta_{x+e_1,\sigma}),
\end{equation}
and 
\begin{multline}
\label{tildeHhop10}
\tilde{H}_{\rm hop,1,0}^{(\Lambda)}
:=\frac{i\kappa}{2}\sum_{\sigma=\uparrow,\downarrow}\sum_{x^{(2)},\ldots,x^{(d)}}
[(\xi_{x_1^0,\sigma}\xi_{x_1^1,\sigma}-\eta_{x_1^0,\sigma}\eta_{x_1^1,\sigma}) 
+(\xi_{x_1^-,\sigma}\xi_{x_1^+,\sigma}-\eta_{x_1^-,\sigma}\eta_{x_1^+,\sigma})]\\
=\frac{i\kappa}{2}\sum_{\sigma=\uparrow,\downarrow}\sum_{x^{(2)},\ldots,x^{(d)}}
[\xi_{x_1^0,\sigma}\vartheta(\xi_{x_1^0,\sigma})+\eta_{x_1^0,\sigma}\vartheta(\eta_{x_1^0,\sigma}) 
+\xi_{x_1^-,\sigma}\vartheta(\xi_{x_1^-,\sigma})+\eta_{x_1^-,\sigma}\vartheta(\eta_{x_1^-,\sigma})],
\end{multline}
where we have written $x_1^0=(0,x^{(2)},\ldots,x^{(d)})$ and 
$x_1^1=(1,x^{(2)},\ldots,x^{(d)})$, and used the relation $\vartheta(\eta_{x,\sigma})=-\eta_{\vartheta(x),\sigma}$. 
One can show 
\begin{equation}
\label{tildeHhop1reflec}
\tilde{H}_{\rm hop,1,\pm}^{(\Lambda)}\in\mathfrak{A}_\pm \quad \mbox{and}\quad 
\vartheta(\tilde{H}_{\rm hop,1,-}^{(\Lambda)})=\tilde{H}_{\rm hop,1,+}^{(\Lambda)}.
\end{equation}

We write  
\begin{equation}
\tilde{H}_{\rm hop}^{(\Lambda)}:=(\tilde{U}_1^{(\Lambda)})^\dagger H_{\rm hop}^{(\Lambda)}\tilde{U}_1^{(\Lambda)}.
\end{equation}
Then, from (\ref{tildeHhop2decomp}), (\ref{tildeHhop1decomp}) and (\ref{tildeHhop1reflec}), 
one notices that this Hamiltonian $\tilde{H}_{\rm hop}^{(\Lambda)}$ 
can be decomposed into three parts, 
\begin{equation}
\label{tildeHhopdecomp}
\tilde{H}_{\rm hop}^{(\Lambda)}=\tilde{H}_{\rm hop,-}^{(\Lambda)}+\tilde{H}_{\rm hop,+}^{(\Lambda)}
+\tilde{H}_{\rm hop,1,0}^{(\Lambda)},
\end{equation}
where the first and second terms satisfy 
\begin{equation}
\vartheta(\tilde{H}_{\rm hop,-}^{(\Lambda)})=\tilde{H}_{\rm hop,+}^{(\Lambda)}.
\end{equation}

In order to deal with the interaction Hamiltonian $H_{\rm int}^{(\Lambda)}(h)$ of (\ref{Hinth}), we define
\begin{equation}
H_{{\rm int},j}^{(\Lambda)}(h)
:=\frac{g}{4}\sum_{x\in\Lambda}[\Gamma_x^{(1)}+\Gamma_{x+e_j}^{(1)}+(-1)^{x^{(1)}+\cdots +x^{(d)}}h_j(x)]^2
-\frac{g}{4}\sum_{x\in\Lambda}[\Gamma_x^{(2)}-\Gamma_{x+e_j}^{(2)}]^2
\end{equation} 
for $j=1,2,\ldots,d$. Note that 
\begin{equation}
(U_{1,j}^{(\Lambda)})^\dagger \Gamma_x^{(n)}U_{1,j}^{(\Lambda)}=-(-1)^{x^{(j)}}\Gamma_x^{(n)}
\end{equation}
and  
\begin{equation}
(U_{1,j}^{(\Lambda)})^\dagger \Gamma_{x+e_i}^{(n)}U_{1,j}^{(\Lambda)}=
\begin{cases}
-(-1)^{x^{(j)}}\Gamma_{x+e_i}^{(n)} & \mbox{for \ } i\ne j;\\
(-1)^{x^{(j)}}\Gamma_{x+e_i}^{(n)} & \mbox{for \ } i=j;
\end{cases}
\end{equation}
for $j=2,3,\ldots,d$ and $n=1,2$. These yield 
\begin{equation}
(U_1^{(\Lambda)})^\dagger \Gamma_x^{(n)}U_1^{(\Lambda)}=(-1)^{d-1}(-1)^{x^{(2)}+\cdots+x^{(d)}}\Gamma_x^{(n)}
\end{equation}
and
\begin{equation}
(U_1^{(\Lambda)})^\dagger \Gamma_{x+e_1}^{(n)}U_1^{(\Lambda)}=(-1)^{d-1}(-1)^{x^{(2)}+\cdots+x^{(d)}}\Gamma_{x+e_1}^{(n)}
\end{equation}
for $n=1,2$; and 
\begin{equation}
(U_1^{(\Lambda)})^\dagger \Gamma_{x+e_i}^{(n)}U_1^{(\Lambda)}=(-1)^{d}(-1)^{x^{(2)}+\cdots+x^{(d)}}\Gamma_{x+e_i}^{(n)}
\end{equation}
for $i=2,3,\ldots,d$ and $n=1,2$. Therefore, we have  
\begin{eqnarray}
(U_1^{(\Lambda)})^\dagger H_{\rm int,1}^{(\Lambda)}(h)U_1^{(\Lambda)}&=&\frac{g}{4}
\sum_{x\in\Lambda}[\Gamma_x^{(1)}+\Gamma_{x+e_1}^{(1)}+(-1)^{d-1}(-1)^{x^{(1)}}h_1(x)]^2\ret
&-&\frac{g}{4}\sum_{x\in\Lambda}[\Gamma_x^{(2)}-\Gamma_{x+e_1}^{(2)}]^2
\end{eqnarray}
and
\begin{eqnarray}
(U_1^{(\Lambda)})^\dagger H_{{\rm int},j}^{(\Lambda)}(h)U_1^{(\Lambda)}&=&\frac{g}{4}\sum_{x\in\Lambda}
[\Gamma_x^{(1)}-\Gamma_{x+e_j}^{(1)}+(-1)^{d-1}(-1)^{x^{(1)}}h_j(x)]^2\ret
&-&\frac{g}{4}\sum_{x\in\Lambda}[\Gamma_x^{(2)}+\Gamma_{x+e_j}^{(2)}]^2
\end{eqnarray}
for $j=2,3,\ldots,d$. Further, we have 
\begin{eqnarray}
\label{tildeU1Hint1}
(\tilde{U}_1^{(\Lambda)})^\dagger H_{\rm int,1}^{(\Lambda)}(h)\tilde{U}_1^{(\Lambda)}&=&\frac{g}{4}
\sum_{x\in\Lambda}[\Gamma_x^{(1)}-\Gamma_{x+e_1}^{(1)}+(-1)^{d-1}(-1)^{x^{(2)}+\cdots+x^{(d)}}h_1(x)]^2\ret
&-&\frac{g}{4}\sum_{x\in\Lambda}[\Gamma_x^{(2)}-\Gamma_{x+e_1}^{(2)}]^2
\end{eqnarray}
and
\begin{eqnarray}
\label{tildeU1Hintj}
(\tilde{U}_1^{(\Lambda)})^\dagger H_{{\rm int},j}^{(\Lambda)}(h)\tilde{U}_1^{(\Lambda)}&=&\frac{g}{4}\sum_{x\in\Lambda}
[\Gamma_x^{(1)}+\Gamma_{x+e_j}^{(1)}+(-1)^{d-1}(-1)^{x^{(2)}+\cdots+x^{(d)}}h_j(x)]^2\ret
&-&\frac{g}{4}\sum_{x\in\Lambda}[\Gamma_x^{(2)}+\Gamma_{x+e_j}^{(2)}]^2,
\end{eqnarray}
for $j=2,3,\ldots,d$, where we have used $\tilde{U}_1^{(\Lambda)}=U_1^{(\Lambda)}U_{\rm odd}^{(\Lambda)}$, and 
\begin{equation}
(U_{\rm odd}^{(\Lambda)})^\dagger \Gamma_x^{(1)}U_{\rm odd}^{(\Lambda)}= 
(-1)^{x^{(1)}+\cdots+x^{(d)}}\Gamma_x^{(1)}
\end{equation}
and 
\begin{equation}
(U_{\rm odd}^{(\Lambda)})^\dagger \Gamma_x^{(2)}U_{\rm odd}^{(\Lambda)}=\Gamma_x^{(2)}.
\end{equation}

We write 
\begin{equation}
\tilde{H}_{\rm int}^{(\Lambda)}(h):=(\tilde{U}_1^{(\Lambda)})^\dagger H_{\rm int}^{(\Lambda)}(h)\tilde{U}_1^{(\Lambda)}.
\end{equation}
Similarly to the case of the hopping Hamiltonian, this Hamiltonian can be decomposed into three parts, 
\begin{equation}
\label{tildeHintdecomp}
\tilde{H}_{\rm int}^{(\Lambda)}(h)=\tilde{H}_{\rm int,-}^{(\Lambda)}(h)+\tilde{H}_{\rm int,+}^{(\Lambda)}(h)
+\tilde{H}_{\rm int,0}^{(\Lambda)}(h),
\end{equation}
where the first two terms satisfy 
\begin{equation}
\tilde{H}_{\rm int,\pm}^{(\Lambda)}(h)\in\mathfrak{A}_\pm
\end{equation}
and the third term is given by 
\begin{eqnarray}
\label{tildeHint0}
\tilde{H}_{\rm int,0}^{(\Lambda)}(h)&:=&\frac{g}{4}\sum_{x^{(2)},\ldots,x^{(d)}}
[\Gamma_{x_1^0}^{(1)}-\Gamma_{x_1^1}^{(1)}+\tilde{h}_1(x_1^0)]^2
+\frac{g}{4}\sum_{x^{(2)},\ldots,x^{(d)}}[\Gamma_{x_1^+}^{(1)}-\Gamma_{x_1^-}^{(1)}+\tilde{h}_1(x_1^+)]^2\ret
&-&\frac{g}{4}\sum_{x^{(2)},\ldots,x^{(d)}}[\Gamma_{x_1^0}^{(2)}-\Gamma_{x_1^1}^{(2)}]^2
-\frac{g}{4}\sum_{x^{(2)},\ldots,x^{(d)}}[\Gamma_{x_1^+}^{(2)}-\Gamma_{x_1^-}^{(2)}]^2,
\end{eqnarray}
where 
\begin{equation}
\tilde{h}_1(x):=(-1)^{d-1}(-1)^{x^{(2)}+\cdots+x^{(d)}}h_1(x).
\end{equation}

In the same way, one has 
\begin{equation}
(\tilde{U}_1^{(\Lambda)})^\dagger O^{(\Lambda)}\tilde{U}_1^{(\Lambda)}=(-1)^{d-1}\sum_{x\in\Lambda_-}(-1)^{x^{(1)}}\Gamma_x^{(2)}
+(-1)^{d-1}\sum_{x\in\Lambda_+}(-1)^{x^{(1)}}\Gamma_x^{(2)}.
\end{equation}
Therefore, one obtains
\begin{equation}
\label{tildeHSBFdecomp}
\tilde{H}_{\rm SBF}^{(\Lambda)}(B):=-B\sum_{x\in\Lambda}(\tilde{U}_1^{(\Lambda)})^\dagger O^{(\Lambda)}\tilde{U}_1^{(\Lambda)} 
=\tilde{H}_{\rm SBF,-}^{(\Lambda)}(B)+\tilde{H}_{\rm SBF,+}^{(\Lambda)}(B),
\end{equation}
where 
\begin{equation}
\tilde{H}_{\rm SBF,\pm}^{(\Lambda)}(B):=(-1)^{d}B\sum_{x\in\Lambda_\pm}(-1)^{x^{(1)}}\Gamma_x^{(2)}.
\end{equation}
The two terms in the right-hand side satisfy 
\begin{equation}
\tilde{H}_{\rm SBF,\pm}^{(\Lambda)}(B)\in\mathfrak{A}_\pm\quad \mbox{and}\quad 
\vartheta(\tilde{H}_{\rm SBF,-}^{(\Lambda)}(B))=\tilde{H}_{\rm SBF,+}^{(\Lambda)}(B),
\end{equation}
where we have used $\vartheta(\Gamma_x^{(2)})=-\Gamma_{\vartheta(x)}^{(2)}$.

{From} (\ref{tildeHhopdecomp}), (\ref{tildeHintdecomp})and (\ref{tildeHSBFdecomp}), we have 
\begin{eqnarray}
\label{tildeHdecomp}
\tilde{H}^{(\Lambda)}(B,h)&:=&(\tilde{U}_1^{(\Lambda)})^\dagger H^{(\Lambda)}(B,h)\tilde{U}_1^{(\Lambda)}\ret
&=&\tilde{H}_-^{(\Lambda)}(B,h)+\tilde{H}_+^{(\Lambda)}(B,h)+\tilde{H}_0^{(\Lambda)}(h),
\end{eqnarray}
where $\tilde{H}_\pm^{(\Lambda)}(B,h)\in\mathfrak{A}_\pm$, and 
\begin{equation}
\tilde{H}_0^{(\Lambda)}(h):=\tilde{H}_{\rm hop,1,0}^{(\Lambda)}+\tilde{H}_{\rm int,0}^{(\Lambda)}(h).
\end{equation}

\begin{pro} The following inequality is valid: 
\label{pro:tildeHBh1d}
\begin{eqnarray}
\label{PRSchwarz}
\left\{{\rm Tr} \exp[-\beta \tilde{H}^{(\Lambda)}(B,h)]\right\}^2
&\le&{\rm Tr} \exp[-\beta(\tilde{H}_-^{(\Lambda)}(B,h)+\vartheta(\tilde{H}_-^{(\Lambda)}(B,h))+\tilde{H}_0^{(\Lambda)}(0))]\ret
&\times&{\rm Tr} \exp[-\beta(\vartheta(\tilde{H}_+^{(\Lambda)}(B,h))+\tilde{H}_+^{(\Lambda)}(B,h)+\tilde{H}_0^{(\Lambda)}(0)  )].\ret
\end{eqnarray}
\end{pro}

\begin{proof}{Proof} We use Lie product formula 
\begin{equation}
e^{\mathcal{A}_1+\mathcal{A}_2+\cdots+\mathcal{A}_n}
=\lim_{M\rightarrow\infty}[(1+\mathcal{A}_1/M)e^{\mathcal{A}_2/M}\cdots e^{\mathcal{A}_n/M}]^M
\end{equation}
for matrices $\mathcal{A}_1, \mathcal{A}_2,\ldots, \mathcal{A}_n$. From the decomposition (\ref{tildeHdecomp}) of 
the Hamiltonian $\tilde{H}^{(\Lambda)}(B,h)$, we have 
\begin{equation}
\exp[-\beta \tilde{H}^{(\Lambda)}(B,h)]=\lim_{M\nearrow \infty}(I_M)^M,
\end{equation}
where 
\begin{equation}
I_M:=\left(1-\frac{\beta}{M}\tilde{H}_{\rm hop,1,0}^{(\Lambda)}\right)J_M
\end{equation}
with 
\begin{eqnarray}
J_M&:=&
\left\{\prod_{(x^{(2)},\ldots,x^{(d)})}
\exp\left[-\frac{\beta g}{4M}(\Gamma_{x_1^0}^{(1)}-\Gamma_{x_1^1}^{(1)}+\tilde{h}_1(x_1^0))^2\right]\right\}
\ret&\times&
\left\{\prod_{(x^{(2)},\ldots,x^{(d)})}\exp\left[-\frac{\beta g}{4M}(\Gamma_{x_1^+}^{(1)}-\Gamma_{x_1^-}^{(1)}
+\tilde{h}_1(x_1^+))^2\right]\right\}\ret 
&\times&
\left\{\prod_{(x^{(2)},\ldots,x^{(d)})}\exp\left[\frac{\beta g}{4M}(\Gamma_{x_1^0}^{(2)}-\Gamma_{x_1^1}^{(2)})^2\right]\right\}
\left\{\prod_{(x^{(2)},\ldots,x^{(d)})}
\exp\left[\frac{\beta g}{4M}(\Gamma_{x_1^+}^{(2)}-\Gamma_{x_1^-}^{(2)})^2\right]\right\}\ret
&\times&\exp\left[-\frac{\beta}{M}\tilde{H}_-^{(\Lambda)}(B,h)\right] \exp\left[-\frac{\beta}{M}\tilde{H}_+^{(\Lambda)}(B,h)\right].
\end{eqnarray}
We have also used the expression (\ref{tildeHint0}) of $\tilde{H}_{\rm int,0}^{(\Lambda)}(h)$. 

By using the operator identity, \cite{DLS}
\begin{equation}
e^{-\hat{D}^2}=\int_{-\infty}^{+\infty}\frac{dk}{\sqrt{4\pi}}\; e^{-k^2/4}e^{ik\hat{D}},
\end{equation}
for a real hermitian $\hat{D}$, and the assumption of the positivity $g>0$ for the coupling constant $g$, 
one has 
\begin{eqnarray}
& &\exp\left[-\frac{\beta g}{4M}(\Gamma_{x_1^0}^{(1)}-\Gamma_{x_1^1}^{(1)}+\tilde{h}_1(x_1^0))^2\right]\ret
&=&\int_{-\infty}^{+\infty}\frac{dk}{\sqrt{4\pi}}\; e^{-k^2/4}
\exp\left[ik\sqrt{\frac{\beta g}{4M}}(\Gamma_{x_1^0}^{(1)}-\Gamma_{x_1^1}^{(1)}+\tilde{h}_1(x_1^0))\right]\ret
&=&\int_{-\infty}^{+\infty}\frac{dk}{\sqrt{4\pi}}\; e^{-k^2/4}
\exp\left[ik\sqrt{\frac{\beta g}{4M}}\Gamma_{x_1^0}^{(1)}\right]
\exp\left[-ik\sqrt{\frac{\beta g}{4M}}\Gamma_{x_1^1}^{(1)}\right]\exp\left[ik\sqrt{\frac{\beta g}{4M}}\tilde{h}_1(x_1^0)\right].\ret
\end{eqnarray}
Similarly, by using \cite{DLS} 
\begin{equation}
e^{\hat{F}^2}=\int_{-\infty}^{+\infty}\frac{dk}{\sqrt{4\pi}}\; e^{-k^2/4}e^{k\hat{F}}
\end{equation}
for a pure imaginary hermitian $\hat{F}$, one obtains 
\begin{eqnarray}
\exp\left[\frac{\beta g}{4M}(\Gamma_{x_1^0}^{(2)}-\Gamma_{x_1^1}^{(2)})^2\right]
&=&\int_{-\infty}^{+\infty}\frac{dk}{\sqrt{4\pi}}\; e^{-k^2/4}
\exp\left[k\sqrt{\frac{\beta g}{4M}}(\Gamma_{x_1^0}^{(2)}-\Gamma_{x_1^1}^{(2)})\right]\ret
&=&\int_{-\infty}^{+\infty}\frac{dk}{\sqrt{4\pi}}\; e^{-k^2/4}
\exp\left[k\sqrt{\frac{\beta g}{4M}}\Gamma_{x_1^0}^{(2)}\right]\exp\left[-k\sqrt{\frac{\beta g}{4M}}\Gamma_{x_1^1}^{(2)}\right].\ret
\end{eqnarray}
These identities yield 
\begin{equation}
J_M=\int d\mu(\tilde{k})\; K_-(\tilde{k},B,h)K_+(\tilde{k},B,h)
\end{equation}
with 
\begin{equation}
\label{K-}
K_-(\tilde{k},B,h):=A_-(\tilde{k})\exp\left[-\frac{\beta}{M}\tilde{H}_-^{(\Lambda)}(B,h)\right],
\end{equation}
\begin{equation}
\label{K+}
K_+(\tilde{k},B,h):=A_+(\tilde{k})\exp\left[-\frac{\beta}{M}\tilde{H}_+^{(\Lambda)}(B,h)\right]
\exp\left[i\Theta(\tilde{k},\tilde{h}_1)\right],
\end{equation}
where we have written 
\begin{eqnarray}
\int d\mu(\tilde{k})
&:=&\prod_{(x^{(2)},\ldots,x^{(d)})}\int_{-\infty}^{+\infty}\frac{dk^{(1)}(x_1^0)}{\sqrt{4\pi}}e^{-[k^{(1)}(x_1^0)]^2/4}
\int_{-\infty}^{+\infty}\frac{dk^{(1)}(x_1^+)}{\sqrt{4\pi}}e^{-[k^{(1)}(x_1^+)]^2/4}\ret
&\times&\int_{-\infty}^{+\infty}\frac{dk^{(2)}(x_1^0)}{\sqrt{4\pi}}e^{-[k^{(2)}(x_1^0)]^2/4}
\int_{-\infty}^{+\infty}\frac{dk^{(2)}(x_1^+)}{\sqrt{4\pi}}e^{-[k^{(2)}(x_1^+)]^2/4},
\end{eqnarray}
\begin{eqnarray}
A_-(\tilde{k})&:=&\left\{\prod_{(x^{(2)},\ldots,x^{(d)})}
\exp\left[ik^{(1)}(x_1^0)\sqrt{\frac{\beta g}{4\pi}}\Gamma_{x_1^0}^{(1)}\right]
\exp\left[-ik^{(1)}(x_1^+)\sqrt{\frac{\beta g}{4\pi}}\Gamma_{x_1^-}^{(1)}\right]\right\}\ret
&\times& \left\{\prod_{(x^{(2)},\ldots,x^{(d)})}
\exp\left[k^{(2)}(x_1^0)\sqrt{\frac{\beta g}{4\pi}}\Gamma_{x_1^0}^{(2)}\right]
\exp\left[-k^{(2)}(x_1^+)\sqrt{\frac{\beta g}{4\pi}}\Gamma_{x_1^-}^{(2)}\right]\right\}, 
\end{eqnarray}
and 
\begin{equation}
\exp\left[i\Theta(\tilde{k},\tilde{h}_1)\right]:=\prod_{(x^{(2)},\ldots,x^{(d)})}
\exp\left[ik^{(1)}(x_1^0)\tilde{h}_1(x_1^0)+ik^{(1)}(x_1^+)\tilde{h}_1(x_1^+)\right]; 
\end{equation}
$A_+(\tilde{k})$ is given by  
\begin{equation}
A_+(\tilde{k})=\vartheta(A_-(\tilde{k}))\in\mathfrak{A}_+.
\end{equation}
Here, we have also used the following relations:  
\begin{equation}
\vartheta(\Gamma_{x_1^0}^{(1)})=\Gamma_{x_1^1}^{(1)},\ \vartheta(\Gamma_{x_1^-}^{(1)})=\Gamma_{x_1^+}^{(1)}, \
\vartheta(\Gamma_{x_1^0}^{(2)})=-\Gamma_{x_1^1}^{(2)},\ \vartheta(\Gamma_{x_1^-}^{(2)})=-\Gamma_{x_1^+}^{(2)}. 
\end{equation}
Clearly, one has 
\begin{equation}
K_\pm(\tilde{k},B,h)\in \mathfrak{A}_\pm.
\end{equation}

Note that 
\begin{eqnarray}
(I_M)^M&=&\int d\mu(\tilde{k}_1)\int d\mu(\tilde{k}_2)\cdots\int d\mu(\tilde{k}_M)
\left(1-\frac{\beta}{M}\tilde{H}_{\rm hop,1,0}^{(\Lambda)}\right)K_-(\tilde{k}_1,B,h)K_+(\tilde{k}_1,B,h)\ret
&\times&\left(1-\frac{\beta}{M}\tilde{H}_{\rm hop,1,0}^{(\Lambda)}\right)K_-(\tilde{k}_2,B,h)K_+(\tilde{k}_2,B,h)\times\cdots \ret
\cdots&\times&\left(1-\frac{\beta}{M}\tilde{H}_{\rm hop,1,0}^{(\Lambda)}\right)K_-(\tilde{k}_M,B,h)K_+(\tilde{k}_M,B,h).
\end{eqnarray}
By using the expression (\ref{tildeHhop10}) of $\tilde{H}_{\rm hop,1,0}^{(\Lambda)}$, 
we expand the integrand so that each term has the following form: 
\begin{eqnarray}
\label{expandeach}
& &\left(\frac{-i\beta\kappa}{2M}\right)^n K_-(\tilde{k}_1,B,h)K_+(\tilde{k}_1,B,h)\cdots 
K_-(\tilde{k}_{\ell_1},B,h)K_+(\tilde{k}_{\ell_1},B,h)\gamma_1\vartheta(\gamma_1)\ret
&\times&K_-(\tilde{k}_{\ell_1+1},B,h)K_+(\tilde{k}_{\ell_1+1},B,h)\cdots K_-(\tilde{k}_{\ell_2},B,h)K_+(\tilde{k}_{\ell_2},B,h)
\gamma_2\vartheta(\gamma_2)\ret
&\times&K_-(\tilde{k}_{\ell_2+1},B,h)K_+(\tilde{k}_{\ell_2+1},B,h)\cdots K_-(\tilde{k}_{\ell_n},B,h)K_+(\tilde{k}_{\ell_n},B,h)
\gamma_n\vartheta(\gamma_n)\ret
&\times&K_-(\tilde{k}_{\ell_n+1},B,h)K_+(\tilde{k}_{\ell_n+1},B,h)\cdots K_-(\tilde{k}_M,B,h)K_+(\tilde{k}_M,B,h),
\end{eqnarray}
where $\gamma_\ell\in \mathfrak{A}_-$ takes $\xi_{x_1^0,\sigma}, \eta_{x_1^0,\sigma}, 
\xi_{x_1^-,\sigma},\eta_{x_1^-,\sigma}$.  
When $n={\rm odd}$, the trace of this term is vanishing \cite{JP}. Therefore, it is sufficient to consider 
the case with $n={\rm even}$. Since $K_\pm(\tilde{k},B,h)$ has even fermion parity and $K_\pm(\tilde{k},B,h)\in\mathfrak{A}_\pm$, 
one has 
\begin{eqnarray}
& &K_-(\tilde{k}_1,B,h)K_+(\tilde{k}_1,B,h)\cdots 
K_-(\tilde{k}_{\ell_1},B,h)K_+(\tilde{k}_{\ell_1},B,h)\gamma_1\vartheta(\gamma_1)\ret
&\times&K_-(\tilde{k}_{\ell_1+1},B,h)K_+(\tilde{k}_{\ell_1+1},B,h)\cdots K_-(\tilde{k}_{\ell_2},B,h)K_+(\tilde{k}_{\ell_2},B,h)
\gamma_2\vartheta(\gamma_2)\ret
&=&K_-(\tilde{k}_1,B,h)\cdots K_-(\tilde{k}_{\ell_1},B,h)\gamma_1
K_+(\tilde{k}_1,B,h)\cdots K_+(\tilde{k}_{\ell_1},B,h)\vartheta(\gamma_1)\ret
&\times&K_-(\tilde{k}_{\ell_1+1},B,h)\cdots K_-(\tilde{k}_{\ell_2},B,h)\gamma_2 
K_+(\tilde{k}_{\ell_1+1},B,h)\cdots K_+(\tilde{k}_{\ell_2},B,h)\vartheta(\gamma_2)\ret
&=&(-1)K_-(\tilde{k}_1,B,h)\cdots K_-(\tilde{k}_{\ell_1},B,h)\gamma_1
K_-(\tilde{k}_{\ell_1+1},B,h)\cdots K_-(\tilde{k}_{\ell_2},B,h)\gamma_2\ret
&\times&K_+(\tilde{k}_1,B,h)\cdots K_+(\tilde{k}_{\ell_1},B,h)\vartheta(\gamma_1)
K_+(\tilde{k}_{\ell_1+1},B,h)\cdots K_+(\tilde{k}_{\ell_2},B,h)\vartheta(\gamma_2),\ret
\end{eqnarray}
where the factor $(-1)$ is obtained by the anti-commutation relation $\vartheta(\gamma_1)\gamma_2=-\gamma_2\vartheta(\gamma_1)$ 
for the Majorana fermions. Clearly, this right-hand side has the form $X_-X_+$ with $X_\pm\in\mathfrak{A}_\pm$. 
Therefore, the term of (\ref{expandeach}) can be written into the form, 
\begin{equation}
\left(\frac{\beta\kappa}{2M}\right)^{2m} X_-(1)X_+(1)X_-(2)X_+(2)\cdots X_-(m)X_+(m),
\end{equation}
where we have written $n=2m$ with the integer $m\ge 0$, and $X_\pm(j)\in\mathfrak{A}_\pm$ for $j=1,2,\ldots,m$. 
Since $X_\pm(j)$ has even fermion parity, one has 
\begin{eqnarray}
& &\left(\frac{\beta\kappa}{2M}\right)^{2m} X_-(1)X_+(1)X_-(2)X_+(2)\cdots X_-(m)X_+(m)\ret
&=&\left(\frac{\beta\kappa}{2M}\right)^{2m} X_-(1)X_-(2)\cdots X_-(m)X_+(1)X_+(2)\cdots X_+(m). 
\end{eqnarray}
Therefore, ${\rm Tr}(I_M)^M$ can be written in the form, 
\begin{equation}
\label{TrIMM}
{\rm Tr}(I_M)^M=\int d\mu(\tilde{k}_1)\int d\mu(\tilde{k}_2)\cdots\int d\mu(\tilde{k}_M)\sum_j {\rm Tr}\; W_-(j)W_+(j),
\end{equation}
where 
\begin{equation}
W_\pm(j):=\left(\frac{\beta\kappa}{2M}\right)^{m_j}X_\pm(1,j)X_\pm(2,j)\cdots X_\pm(m_j,j)
\end{equation}
with integer $m_j\ge 0$ and operator $X_\pm(i,j)\in\mathfrak{A}_\pm$ for $i=1,2,\ldots,m_j$.
 
Let $\mathcal{A}_j=\mathcal{A}_j(\tilde{k}_1,\tilde{k}_2,\ldots,\tilde{k}_M)\in\mathfrak{A}_-$ and  
$\mathcal{B}_j=\mathcal{B}_j(\tilde{k}_1,\tilde{k}_2,\ldots,\tilde{k}_M)\in\mathfrak{A}_-$ be 
two sets $\{\mathcal{A}_j\}$ and $\{\mathcal{B}_j\}$ of operator-valued 
functions of $\tilde{k}_1,\tilde{k}_2,\ldots,\tilde{k}_M$. 
By relying on \cite{JP} the positivity ${\rm Tr}\; \mathcal{A}\vartheta(\mathcal{A})\ge 0$ for $\mathcal{A}\in\mathfrak{A}_-$, 
we define an inner product by 
\begin{equation}
\langle \! \langle \{\mathcal{A}_j\},\{\mathcal{B}_j\}\rangle\! \rangle:=
\int d\mu(\tilde{k}_1)\int d\mu(\tilde{k}_2)\cdots\int d\mu(\tilde{k}_M)\sum_j {\rm Tr}\; \mathcal{A}_j\vartheta(\mathcal{B}_j). 
\end{equation}
This yields Schwarz inequality, 
\begin{equation}
\left|\langle \! \langle \{\mathcal{A}_j\},\{\mathcal{B}_j\}\rangle\! \rangle\right|^2
\le \langle \! \langle \{\mathcal{A}_j\},\{\mathcal{A}_j\}\rangle\! \rangle
\langle \! \langle \{\mathcal{B}_j\},\{\mathcal{B}_j\}\rangle\! \rangle.
\end{equation}
Since $W_+(j)=\vartheta(\vartheta(W_+(j)))$ with $\vartheta(W_+(j))\in\mathfrak{A}_-$, the application of 
Schwarz inequality to the right-hand side of (\ref{TrIMM}) yields the desired result (\ref{PRSchwarz}). 
More precisely, the difference between $W_-(j)$ and $W_+(j)$ comes from (\ref{K-}) and (\ref{K+}). 
The two Hamiltonians $\tilde{H}_\pm^{(\Lambda)}(B,h)$ are mapped onto the opposite side by the reflection $\vartheta$, 
and the phase $\Theta$ in the right-hand side of (\ref{K+}) is vanishing. 
\end{proof}

Combining this Proposition~\ref{pro:tildeHBh1d} with the expressions (\ref{tildeU1Hint1}) and (\ref{tildeU1Hintj}), 
one has: 

\begin{coro} For any given set of $d$ real-valued functions $h=\{h_1,h_2,\ldots,h_m\}$, the following inequality is valid: 
\begin{equation}
\label{PRSchwarzCoro}
\left\{{\rm Tr} \exp[-\beta \tilde{H}^{(\Lambda)}(B,h)]\right\}^2
\le{\rm Tr} \exp[-\beta\tilde{H}^{(\Lambda)}(B,h^-)]\times
{\rm Tr} \exp[-\beta\tilde{H}^{(\Lambda)}(B,h^+)],
\end{equation}
where the functions $h^+=\{h_1^+,h_2^+,\ldots,h_d^+\}$ are given by  
\begin{equation}
h_1^+(x)=
\begin{cases}
h_1(x), & \mbox{for \ } x\in \Lambda_+ \ \mbox{and} \ x^{(1)}\ne L;\\
-h_1(\vartheta(x+e_1)), & \mbox{for \ }x\in \Lambda_- \ \mbox{and} \ x^{(1)}\ne 0,
\end{cases}
\end{equation}
\begin{equation}
h_1^+(x_1^0)=h_1^+(x_1^+)=0,
\end{equation}
and when $m=2,3,\ldots,d$, 
\begin{equation}
h_m^+(x)=
\begin{cases}
h_m(x), & \mbox{for \ } x\in\Lambda_+;\\
h_m(\vartheta(x)), & \mbox{for \ } x\in\Lambda_-.
\end{cases}
\end{equation}
Similarly, the functions $h^-=\{h_1^-,h_2^-,\ldots,h_d^-\}$ are given by 
\begin{equation}
h_1^-(x)=
\begin{cases}
h_1(x) & \mbox{for \ } x\in \Lambda_- \ \mbox{and \ } x^{(1)}\ne 0;\\
-h_1(\vartheta(x+e_1)) & \mbox{for \ } x\in \Lambda_+ \ \mbox{and} \ x^{(1)}\ne L,
\end{cases}
\end{equation}
\begin{equation}
h_1^-(x_1^0)=h_1^-(x_1^+)=0,
\end{equation}
and if $m=2,3,\ldots,d$, 
\begin{equation}
h_m^-(x)=
\begin{cases}
h_m(x), & \mbox{for \ } x\in\Lambda_-;\\
h_m(\vartheta(x)), & \mbox{for \ } x\in\Lambda_+.  
\end{cases}
\end{equation}
\end{coro}

Further, from the definition (\ref{tildeHdecomp}) of $\tilde{H}^{(\Lambda)}(B,h)$, one has 
\begin{equation}
\label{PRSchwarzCoroP}
\left\{{\rm Tr} \exp[-\beta {H}^{(\Lambda)}(B,h)]\right\}^2
\le{\rm Tr} \exp[-\beta{H}^{(\Lambda)}(B,h^-)]\times
{\rm Tr} \exp[-\beta{H}^{(\Lambda)}(B,h^+)]
\end{equation}
for the original Hamiltonian $H^{(\Lambda)}(B,h)$. 
In the argument in the proof of Proposition~\ref{pro:tildeHBh1d}, the anti-periodic boundary condition 
in the Hamiltonian $H_{\rm hop,1}^{(\Lambda)}$ of (\ref{Hhop1}) is crucial. 
But, as shown in Sec.~\ref{GaugeTransBC}, we can change the locations of the bonds with the opposite sign of the hopping amplitudes 
due to the anti-periodic boundary conditions by using a gauge transformation. 
Further, as shown in Sec.~\ref{GaugeTransHA}, we can also interchange the roles of the hopping amplitudes 
in the $x^{(1)}$ and $x^{(j)}$ directions for all $j=2,3,\ldots,d$. 
Combining these observations with the argument in the proof of Theorem~4.2 in \cite{DLS}, we obtain:  

\begin{theorem}
\label{theorem:GaussDomi}
The following bound is valid: 
\begin{equation}
\label{GaussDomi}
{\rm Tr}\; \exp[-\beta {H}^{(\Lambda)}(B,h)]\le {\rm Tr}\; \exp[-\beta {H}^{(\Lambda)}(B,0)]
\end{equation}
for any real-valued functions $h=\{h_1,h_2,\ldots,h_d\}$. 
\end{theorem}

Note that 
$$
\Gamma_{x_1^0}^{(1)}\Gamma_{x_1^1}^{(1)}=(a_{x_1^0,\uparrow}^\dagger a_{x_1^0,\downarrow}^\dagger 
+a_{x_1^0,\downarrow}a_{x_1^0,\uparrow})(a_{x_1^1,\uparrow}^\dagger a_{x_1^1,\downarrow}^\dagger 
+a_{x_1^1,\downarrow}a_{x_1^1,\uparrow})=\Gamma_{x_1^0}^{(1)}\vartheta(\Gamma_{x_1^0}^{(1)}),
$$
where $x_1^0=(0,x^{(2)},\ldots,x^{(d)})$ and $x_1^1=(1,x^{(2)},\ldots,x^{(d)})$. 
Therefore, the argument in the proof of Proposition~\ref{pro:tildeHBh1d} also yields 
$$
{\rm Tr}\; \left[\Gamma_{x_1^0}^{(1)}\Gamma_{x_1^1}^{(1)}e^{-\beta\tilde{H}^{(\Lambda)}(B,0)}\right]\ge 0.
$$
This implies 
\begin{equation}
\label{GammaGammanegative}
\langle \Gamma_x^{(1)}\Gamma_y^{(1)}\rangle_{\beta,B}^{(\Lambda)}\le 0
\end{equation}
for any $x,y$ satisfying $|x-y|=1$. Similarly, one has 
\begin{equation}
\langle \Gamma_x^{(2)}\Gamma_y^{(2)}\rangle_{\beta,B}^{(\Lambda)}\le 0
\end{equation}
for any $x,y$ satisfying $|x-y|=1$. 

\Section{Long-range order}

By using the Gaussian domination bound (\ref{GaussDomi}), we prove the existence of the long-range order of 
the superconductivity in this section.  

We write
\begin{equation}
\label{breveHam}
\breve{H}^{(\Lambda)}(B,h):=(U_{\rm odd}^{(\Lambda)})^\dagger H^{(\Lambda)}(B,h)U_{\rm odd}^{(\Lambda)}.
\end{equation}
Then, the bound (\ref{GaussDomi}) implies 
\begin{equation}
\label{GaussDomiodd}
{\rm Tr}\; \exp[-\beta \breve{H}^{(\Lambda)}(B,h)]\le {\rm Tr}\; \exp[-\beta \breve{H}^{(\Lambda)}(B,0)]. 
\end{equation}
For two operators, $\mathcal{A}$ and $\mathcal{B}$, the Duhamel two-point function is defined by 
\begin{equation}
(\!(\mathcal{A},\mathcal{B})\!)_{\beta,B}^{(\Lambda)}:=\frac{1}{\breve{Z}^{(\Lambda)}(B,0)}
\int_0^1 ds\; {\rm Tr} \left[e^{-s\beta \breve{H}^{(\Lambda)}(B,0)}\mathcal{A}e^{-(1-s)\beta\breve{H}^{(\Lambda)}(B,0)}
\mathcal{B}\right]
\end{equation}
with the partition function $\breve{Z}^{(\Lambda)}(B,0)={\rm Tr}\; e^{-\beta\breve{H}^{(\Lambda)}(B,0)}$. 
The bound (\ref{GaussDomiodd}) yields \cite{DLS} 
\begin{equation}
\label{preIRB}
(\!(\Gamma^{(1)}[\overline{\sum_m\partial_m h_m}],\Gamma^{(1)}[\sum_m\partial_m h_m])\!)_{\beta,B}^{(\Lambda)}
\le \frac{1}{\beta g}\sum_m \sum_{x\in\Lambda}\left|h_m(x)\right|^2 
\end{equation}
for $g>0$, where $\overline{(\cdots)}$ stands for complex conjugate, and we have written 
$$
\partial_m h_m(x)=h_m(x)-h_m(x-e_m)
$$
and 
$$
\Gamma^{(1)}[\psi]=\sum_{x\in\Lambda}\Gamma_x^{(1)}\psi(x)
$$
for a function $\psi$ on $\Lambda$. 
This bound has been proved for real-valued functions $h_m$, but it extends to complex-valued functions $h_m$ 
as proved in \cite{DLS}. The bound (\ref{preIRB}) is obtained as follows: 
By using
\begin{equation}
\label{U2transGamma}
(U_{\rm odd}^{(\Lambda)})^\dagger \Gamma_x^{(1)}U_{\rm odd}^{(\Lambda)}=(-1)^{x^{(1)}+\cdots+x^{(d)}}\Gamma_x^{(1)}
\quad \mbox{and} \quad (U_{\rm odd}^{(\Lambda)})^\dagger \Gamma_x^{(2)}U_{\rm odd}^{(\Lambda)}=\Gamma_x^{(2)},
\end{equation}
one has 
\begin{eqnarray}
(U_{\rm odd}^{(\Lambda)})^\dagger H_{\rm int}^{(\Lambda)}(h)U_{\rm odd}^{(\Lambda)}
&=&\frac{g}{4}\sum_{x\in\Lambda}\sum_{m=1}^d [\Gamma_x^{(1)}-\Gamma_{x+e_m}^{(1)}+h_m(x)]^2
-\frac{g}{4}\sum_{x\in\Lambda}\sum_{m=1}^d [\Gamma_x^{(2)}-\Gamma_{x+e_m}^{(2)}]^2\ret
&-&\frac{dg}{2}\sum_{x\in\Lambda}\left\{[\Gamma_x^{(1)}]^2-[\Gamma_x^{(2)}]^2\right\}
\end{eqnarray}
{from} the expression (\ref{Hinth}) of $H_{\rm int}^{(\Lambda)}(h)$. 
The summand of the first sum in the right-hand side is written 
\begin{equation}
[\Gamma_x^{(1)}-\Gamma_{x+e_m}^{(1)}+h_m(x)]^2
=[\Gamma_x^{(1)}-\Gamma_{x+e_m}^{(1)}]^2 +2(\Gamma_x^{(1)}-\Gamma_{x+e_m}^{(1)})h_m(x)
+[h_m(x)]^2. 
\end{equation}
Note that 
$$
\sum_{x\in\Lambda}\Gamma_{x+e_m}^{(1)}h_m(x)=\sum_{x\in\Lambda}\Gamma_x^{(1)}h_m(x-e_m). 
$$
This yields 
$$
\sum_{x\in\Lambda}(\Gamma_x^{(1)}-\Gamma_{x+e_m}^{(1)})h_m(x)
=\sum_{x\in\Lambda}\Gamma_x^{(1)}[h_m(x)-h_m(x-e_m)].
$$
{From} these observations and (\ref{GaussDomi}), one has the bound (\ref{preIRB}).  
 
We choose 
\begin{equation}
h_m(x)=\frac{1}{\sqrt{|\Lambda|}}[e^{ip\cdot (x+e_m)}-e^{ip\cdot x}],
\end{equation}
where $p=(p^{(1)},p^{(2)},\ldots,p^{(d)})\in (-\pi,\pi]^d$ is the wavevector. 
Then, one has 
$$
\partial_m h_m(x)=h_m(x)-h_m(x-e_m)=\frac{-2}{\sqrt{|\Lambda|}}e^{ip\cdot x}(1-\cos p^{(m)})
$$
and 
$$
\sum_{x\in\Lambda}\sum_m |h_m(x)|^2=2E_p,
$$
where we have written 
$$
E_p=\sum_m (1-\cos p^{(m)}).
$$
Substituting these into the bound (\ref{preIRB}), one obtain
\begin{equation}
(\!(\hat{\Gamma}_{p}^{(1)},\hat{\Gamma}_{-p}^{(1)})\!)_{\beta,B}^{(\Lambda)}\le \frac{1}{2\beta g E_p},
\end{equation}
where we have written 
\begin{equation}
\label{hatGammap1}
\hat{\Gamma}_p^{(1)}:=\frac{1}{\sqrt{|\Lambda|}}\sum_{x\in\Lambda}\Gamma_x^{(1)}e^{-ip\cdot x}.
\end{equation}
By using (\ref{breveHam}) and (\ref{U2transGamma}), we have 
\begin{equation}
\label{bpupbound}
(\hat{\Gamma}_{p}^{(1)},\hat{\Gamma}_{-p}^{(1)})\le \frac{1}{2\beta g E_{p+Q}},
\end{equation}
where $Q=(\pi,\ldots,\pi)$ and  
\begin{equation}
(\mathcal{A},\mathcal{B})_{\beta,B}^{(\Lambda)}:=\frac{1}{Z_{\beta,B}^{(\Lambda)}}
\int_0^1 ds\; {\rm Tr} \left[e^{-s\beta {H}^{(\Lambda)}(B,0)}\mathcal{A}e^{-(1-s)\beta{H}^{(\Lambda)}(B,0)}
\mathcal{B}\right]
\end{equation}
with the partition function $Z_{\beta,B}^{(\Lambda)}={\rm Tr}\; e^{-\beta {H}^{(\Lambda)}(B,0)}$ 
for two operators, $\mathcal{A}$ and $\mathcal{B}$. 

We write
\begin{equation}
\label{gp}
\mathfrak{g}_p=\frac{1}{2}\left[\langle \hat{\Gamma}_p^{(1)}\hat{\Gamma}_{-p}^{(1)}\rangle_{\beta,B}^{(\Lambda)}
+\langle \hat{\Gamma}_{-p}^{(1)}\hat{\Gamma}_p^{(1)}\rangle_{\beta,B}^{(\Lambda)}\right],
\end{equation} 
\begin{equation}
\label{bp}
\mathfrak{b}_p=(\hat{\Gamma}_{p}^{(1)},\hat{\Gamma}_{-p}^{(1)})_{\beta,B}^{(\Lambda)},
\end{equation}
and 
\begin{equation}
\label{cp}
\mathfrak{c}_p=\langle [\hat{\Gamma}_{-p}^{(1)},[H^{(\Lambda)}(B,0),\hat{\Gamma}_p^{(1)}]]\rangle_{\beta,B}^{(\Lambda)},
\end{equation}
where $\langle \cdots \rangle_{\beta,B}^{(\Lambda)}$ is the thermal expectation value given by (\ref{TEV}). 
Then, the following inequality is valid: \cite{FB,DLS,JNFP} 
\begin{equation}
\label{DLSboundO}
\mathfrak{g}_p\le \frac{1}{2}\left[\mathfrak{b}_p
+\sqrt{\mathfrak{b}_p^2 +\beta \mathfrak{b}_p\mathfrak{c}_p}\right]
\le \frac{1}{2\beta g E_{p+Q}}+\frac{1}{2}\sqrt{ \frac{1}{2gE_{p+Q}} \mathfrak{c}_p}, 
\end{equation}
where we have used (\ref{bpupbound}) and (\ref{bp}), and the inequality $\sqrt{a+b}\le\sqrt{a}+\sqrt{b}$ 
for $a,b\ge 0$. 

In order to prove the existence of the long-range order of the superconductivity, we consider 
the case with $B=0$ in the dimensions $d\ge 3$. 
Following \cite{KLS,KLS2}, we have 
\begin{multline}
\label{DLSbound}
\frac{1}{|\Lambda|}\sum_p \langle \hat{\Gamma}_p^{(1)}\hat{\Gamma}_{-p}^{(1)}\rangle_{\beta,0}^{(\Lambda)}
\Bigl[-\frac{1}{d}\sum_{i=1}^d\cos p^{(i)}\Bigr] 
\le \frac{1}{|\Lambda|}\sum_{p\ne Q}\frac{1}{2d\beta g E_{p+Q}} \Bigl\{-\sum_{i=1}^d\cos p^{(i)}\Bigr\}_+ \\
+\frac{1}{|\Lambda|}\sum_{p\ne Q}\sqrt{\frac{\mathfrak{c}_p}{8d^2gE_{p+Q}}} \Bigl\{-\sum_{i=1}^d\cos p^{(i)}\Bigr\}_+
+\frac{1}{|\Lambda|}\langle \hat{\Gamma}_Q^{(1)}\hat{\Gamma}_{Q}^{(1)}\rangle_{\beta,0}^{(\Lambda)}
\end{multline}
from the inequality (\ref{DLSboundO}) and the definition (\ref{gp}) of $\mathfrak{g}_p$, 
where $\{a\}_+=\max\{a,0\}$. The last term in the right-hand side is equal to the long-range order. 
In fact, the long-range order is defined by 
\begin{equation}
m_{\rm LRO}^{(\Lambda)}:=\left[\frac{1}{|\Lambda|^2}\sum_{x,y\in\Lambda}
 \langle \Gamma_x^{(1)}\Gamma_y^{(1)}\rangle_{\beta,0}^{(\Lambda)}\times 
(-1)^{x^{(1)}+\cdots+x^{(d)}}\times (-1)^{y^{(1)}+\cdots+y^{(d)}}  \right]^{1/2},
\end{equation}
and 
\begin{equation}
(m_{\rm LRO}^{(\Lambda)})^2=\frac{1}{|\Lambda|}\langle \hat{\Gamma}_Q^{(1)}\hat{\Gamma}_{Q}^{(1)}\rangle_{\beta,0}^{(\Lambda)}.
\end{equation}
The left-hand side of (\ref{DLSbound}) can be written 
\begin{equation}
\label{sumruleKLS}
-\frac{1}{d|\Lambda|}\sum_p \sum_{i=1}^d \langle \hat{\Gamma}_p^{(1)}\hat{\Gamma}_{-p}^{(1)}\rangle_{\beta,0}^{(\Lambda)}\cos p^{(i)}
=-\frac{1}{2|\Lambda|}\sum_{x\in\Lambda}\langle \Gamma_x^{(1)}\Gamma_{x-e_1}^{(1)}
+\Gamma_x^{(1)}\Gamma_{x+e_1}^{(1)}\rangle_{\beta,0}^{(\Lambda)}
\end{equation}
{from} the expression (\ref{hatGammap1}) and (\ref{direcIndep}). 
In the present case, the first sum in the right-hand side of (\ref{DLSbound}) gives the small contribution for a large $\beta$. 
For this contribution, we write 
\begin{equation}
\delta(\beta)=\frac{1}{|\Lambda|}\sum_{p\ne Q}\frac{1}{2d\beta g E_{p+Q}}\Bigl\{-\sum_{i=1}^d\cos p^{(i)}\Bigr\}_+.
\end{equation}
By using the Schwarz inequality and the positivity of $\mathfrak{c}_p$, 
the second term in the right-hand side of (\ref{DLSbound}) can be estimated as 
\begin{multline}
\frac{1}{|\Lambda|}\sum_{p\ne Q}\sqrt{\frac{\mathfrak{c}_p}{8d^2gE_{p+Q}}} \Bigl\{-\sum_i\cos p^{(i)}\Bigr\}_+\\
\qquad \qquad \qquad \quad \quad 
\le \frac{1}{2d\sqrt{2g}}\sqrt{\frac{1}{|\Lambda|}\sum_{p\ne Q}\frac{(\{-\sum_i\cos p^{(i)}\}_+)^2}{E_{p+Q}}}
\times \sqrt{\frac{1}{|\Lambda|}\sum_{p\ne Q}\mathfrak{c}_p }\\
\le \frac{1}{2d\sqrt{2g}}\sqrt{\frac{1}{|\Lambda|}\sum_{p\ne Q}\frac{(\{-\sum_i\cos p^{(i)}\}_+)^2}{E_{p+Q}}}
\times \sqrt{\frac{1}{|\Lambda|}\sum_{p}\mathfrak{c}_p }.
\end{multline}
{From} these observations, one has 
\begin{equation}
\label{KLSbound}
-\frac{1}{2|\Lambda|}\sum_{x\in\Lambda}\langle \Gamma_x^{(1)}\Gamma_{x-e_1}^{(1)}
+\Gamma_x^{(1)}\Gamma_{x+e_1}^{(1)}\rangle_{\beta,0}^{(\Lambda)}
\le\delta(\beta)+\frac{I_d}{2\sqrt{2dg}}\times \sqrt{\frac{1}{|\Lambda|}\sum_{p}\mathfrak{c}_p }+
(m_{\rm LRO}^{(\Lambda)})^2,
\end{equation}
where we have written 
\begin{equation}
\label{Id}
I_d:=\sqrt{\frac{1}{d|\Lambda|}\sum_{p\ne Q}\frac{(\{-\sum_i\cos p^{(i)}\}_+)^2}{E_{p+Q}}}.
\end{equation}

Next, let us calculate the double commutator in $\mathfrak{c}_p$ of (\ref{cp}). 
Since $H^{(\Lambda)}(0,0)=H^{(\Lambda)}(0)=H_{\rm hop}^{(\Lambda)}+H_{\rm int}^{(\Lambda)}$, 
we have 
\begin{equation}
\label{GpHdoubleCommu}
[\hat{\Gamma}_{-p}^{(1)},[H^{(\Lambda)}(0,0),\hat{\Gamma}_p^{(1)}]]
=[\hat{\Gamma}_{-p}^{(1)},[H_{\rm hop}^{(\Lambda)},\hat{\Gamma}_p^{(1)}]]
+[\hat{\Gamma}_{-p}^{(1)},[H_{\rm int}^{(\Lambda)},\hat{\Gamma}_p^{(1)}]].
\end{equation}
Note that 
$$
[\Gamma_x^{(1)},a_{x,\uparrow}^\dagger a_{y,\uparrow}-a_{y,\uparrow}^\dagger a_{x,\uparrow}]
=a_{y,\uparrow}^\dagger a_{x,\downarrow}^\dagger +a_{x,\downarrow}a_{y,\uparrow},
$$
$$
[\Gamma_x^{(1)},a_{y,\uparrow}^\dagger a_{x,\downarrow}^\dagger +a_{x,\downarrow}a_{y,\uparrow}]
=a_{x,\uparrow}^\dagger a_{y,\uparrow}-a_{y,\uparrow}^\dagger a_{x,\uparrow}
$$
and 
$$
\left\Vert a_{x,\uparrow}^\dagger a_{y,\uparrow}-a_{y,\uparrow}^\dagger a_{x,\uparrow}\right\Vert=1.
$$
Combining these with the expressions (\ref{Hhop}), (\ref{Hhop1}), (\ref{Hhop2}) and (\ref{hatGammap1}), 
the first term in the right-hand side of (\ref{GpHdoubleCommu}) can be estimated as 
\begin{equation}
\label{DcommuGammaHhop}
\left\Vert [\hat{\Gamma}_{-p}^{(1)},[H_{\rm hop}^{(\Lambda)},\hat{\Gamma}_p^{(1)}]]\right\Vert
\le 8d|\kappa|.
\end{equation}
Therefore, when $|\kappa|/g$ is sufficiently small, the corresponding contribution is small.    
In the following, we will consider this case. 

Next, consider the second term in the right-hand side of (\ref{GpHdoubleCommu}). From  
the expression (\ref{HintGamma12}) of $H_{\rm int}^{(\Lambda)}$, one has 
\begin{eqnarray}
[H_{\rm int}^{(\Lambda)},\hat{\Gamma}_p^{(1)}]&=&
\frac{1}{\sqrt{|\Lambda|}}\sum_{x\in\Lambda}e^{-ip\cdot x}[H_{\rm int}^{(\Lambda)},\Gamma_x^{(1)}]\ret 
&=&\frac{g}{2}\frac{1}{\sqrt{|\Lambda|}}\sum_{x\in\Lambda}\sum_{y\in\Lambda}\sum_{m=1}^d e^{-ip\cdot x}
[\Gamma_y^{(1)}\Gamma_{y+e_m}^{(1)}+\Gamma_y^{(2)}\Gamma_{y+e_m}^{(2)},\Gamma_x^{(1)}]\ret
&=&-ig\frac{1}{\sqrt{|\Lambda|}}\sum_{x\in\Lambda}\sum_{m=1}^d e^{-ip\cdot x}
\Gamma_x^{(3)}(\Gamma_{x+e_m}^{(2)}+\Gamma_{x-e_m}^{(2)}), 
\end{eqnarray}
where we have used the commutation relation $[\Gamma_x^{(1)},\Gamma_x^{(2)}]=2i\Gamma_x^{(3)}$. 
Further, 
\begin{eqnarray}
[\hat{\Gamma}_{-p}^{(1)},[H_{\rm int}^{(\Lambda)},\hat{\Gamma}_p^{(1)}]]&=&
-ig \frac{1}{|\Lambda|}\sum_{z\in\Lambda}\sum_{x\in\Lambda}\sum_{m=1}^d e^{ip\cdot z}e^{-ip\cdot x}
[\Gamma_z^{(1)},\Gamma_x^{(3)}(\Gamma_{x+e_m}^{(2)}+\Gamma_{x-e_m}^{(2)})]\ret
&=&-2g\frac{1}{|\Lambda|}\sum_{x\in\Lambda}\sum_{m=1}^d (\Gamma_x^{(2)}\Gamma_{x+e_m}^{(2)}+\Gamma_x^{(2)}\Gamma_{x-e_m}^{(2)})\ret
&+&2g \frac{1}{|\Lambda|}\sum_{x\in\Lambda}\sum_{m=1}^d 
(\Gamma_x^{(3)}\Gamma_{x+e_m}^{(3)}e^{ip^{(m)}}+\Gamma_x^{(3)}\Gamma_{x-e_m}^{(3)}e^{-ip^{(m)}}),
\end{eqnarray}
where we have used the above commutation relation and $[\Gamma_x^{(3)},\Gamma_x^{(1)}]=2i\Gamma_x^{(2)}$. 
This yields 
\begin{equation}
\frac{1}{|\Lambda|}\sum_p \langle [\hat{\Gamma}_{-p}^{(1)},[H_{\rm int}^{(\Lambda)},\hat{\Gamma}_p^{(1)}]]\rangle_{\beta,0}^{(\Lambda)}
=-2g\frac{1}{|\Lambda|}\sum_{x\in\Lambda}\sum_{m=1}^d 
\langle \Gamma_x^{(2)}\Gamma_{x+e_m}^{(2)}+\Gamma_x^{(2)}\Gamma_{x-e_m}^{(2)}\rangle_{\beta,0}^{(\Lambda)}. 
\end{equation}
By combining this with (\ref{DcommuGammaHhop}), we obtain 
\begin{equation}
\label{cpbound}
\frac{1}{|\Lambda|}\sum_p \mathfrak{c}_p\le 
8d|\kappa|-2g\frac{1}{|\Lambda|}\sum_{x\in\Lambda}\sum_{m=1}^d 
\langle \Gamma_x^{(2)}\Gamma_{x+e_m}^{(2)}+\Gamma_x^{(2)}\Gamma_{x-e_m}^{(2)}\rangle_{\beta,0}^{(\Lambda)}. 
\end{equation}
By using (\ref{direcIndep}) and the rotation given by (\ref{U1rotation}) 
in Appendix~\ref{appendix:U(1)rotation}, the sum in the right-hand side can be written  
\begin{equation}
\sum_{x\in\Lambda}\sum_{m=1}^d 
\langle \Gamma_x^{(2)}\Gamma_{x+e_m}^{(2)}+\Gamma_x^{(2)}\Gamma_{x-e_m}^{(2)}\rangle_{\beta,0}^{(\Lambda)}
=d \sum_{x\in\Lambda}\langle \Gamma_x^{(1)}\Gamma_{x+e_1}^{(1)}+\Gamma_x^{(1)}\Gamma_{x-e_1}^{(1)}\rangle_{\beta,0}^{(\Lambda)},
\end{equation}
where we have used the rotational symmetry of the Hamiltonian $H^{(\Lambda)}(0)$ with $B=0$. 
We write 
\begin{equation}
\label{E1}
\mathcal{E}_1^{(\Lambda)}:=-\frac{1}{2|\Lambda|}
\sum_{x\in\Lambda}\langle \Gamma_x^{(1)}\Gamma_{x+e_1}^{(1)}+\Gamma_x^{(1)}\Gamma_{x-e_1}^{(1)}\rangle_{\beta,0}^{(\Lambda)}.
\end{equation}
{From} the inequality (\ref{GammaGammanegative}), one has $\mathcal{E}_1^{(\Lambda)}\ge 0$. 
Then, the above inequality (\ref{cpbound}) is written  
\begin{equation}
\label{sumcp}
\frac{1}{|\Lambda|}\sum_p \mathfrak{c}_p\le 
8d|\kappa|+4dg\mathcal{E}_1^{(\Lambda)}.
\end{equation}
Substituting this into (\ref{KLSbound}), we obtain 
\begin{equation}
\mathcal{E}_1^{(\Lambda)}
\le\delta(\beta)+I_d\sqrt{{|\kappa|}/{g}}+I_d\sqrt{{\mathcal{E}_1^{(\Lambda)}}/{2}}
+(m_{\rm LRO}^{(\Lambda)})^2,
\end{equation}
where we have used $\sqrt{a+b}\le\sqrt{a}+\sqrt{b}$ for $a,b\ge 0$. Further, this can be written 
\begin{equation}
\sqrt{\mathcal{E}_1}(\sqrt{\mathcal{E}_1}-I_d/\sqrt{2})-\delta(\beta)-I_d\sqrt{{|\kappa|}/{g}}
\le (m_{\rm LRO})^2,
\end{equation}
where we have written 
$$
\mathcal{E}_1:=\lim_{\Lambda\nearrow\ze^d}\mathcal{E}_1^{(\Lambda)}
$$
and 
$$
m_{\rm LRO}:=\lim_{\Lambda\nearrow\ze^d}m_{\rm LRO}^{(\Lambda)}. 
$$
The quantity $\mathcal{E}_1$ satisfies 
\begin{equation}
\label{E1bound}
\mathcal{E}_1\ge \frac{1}{2} -\frac{\tilde{\delta}(\beta)}{dg}-\frac{|\kappa|}{g}, 
\end{equation}
where $\tilde{\delta}(\beta)$ becomes small for a large $\beta$. This inequality is proved in Appendix~\ref{Appendix:meanenergy}. 
Consequently, we obtain the following result: 
When the constant $I_d$, which depends on the dimension $d$, satisfies $I_d<1$ in the infinite-volume limit, 
the long-range order exists for a large $\beta$ and a small $|\kappa|/g$.   
Our numerical computations show $I_3=0.68\cdots$ and $I_4=0.44\cdots$. 
Thus, the long-range order exists in the dimensions $d=3,4$.  

In order to prove the existence of the long-range order in the dimension $d\ge 5$, we need to estimate 
$I_d$ of (\ref{Id}). {From} the expression (\ref{Id}) of $I_d$, one has 
\begin{equation}
(I_d)^2=\int dp\; \frac{\Bigl[\{-\frac{1}{d}\sum_{m=1}^d \cos p^{(m)}\}_+\Bigr]^2}{1+\frac{1}{d}\sum_{m=1}^d \cos p^{(m)}    }
\end{equation}
in the infinite-volume limit, where we have written 
$$
\int dp:=\frac{1}{(2\pi)^d}\int_{-\pi}^{+\pi} dp^{(1)}\int_{-\pi}^{+\pi} dp^{(2)}\cdots \int_{-\pi}^{+\pi} dp^{(d)}.
$$
Following \cite{KLS2}, we prove $I_d\le I_3$ for all $d\ge 4$. 

Consider a function $F(s)$ of $s\in[0,1]$ which is given by 
$$
F(s)=\frac{s^2}{1-s}.
$$
One can show that this function $F$ is monotone increasing and convex. 
We write 
$$
Y(p):=-\frac{1}{d}\sum_{m=1}^d \cos p^{(m)}.
$$ 
Then, $(I_d)^2$ can be written as 
\begin{equation}
(I_d)^2=\int dp\; F(\{Y(p)\}_+). 
\end{equation}
We also introduce 
$$
Y_{ijk}(p)=-\frac{1}{3}(\cos p^{(i)}+\cos p^{(j)}+\cos p^{(k)})
$$
for $\{i,j,k\; |\; i\ne j,j\ne k, k\ne i\}\subset\{1,2,\ldots,d\}$. Then, one has 
$$
Y(p)=\frac{1}{ _dC_3}\sum_{\{i,j,k\}}Y_{ijk}(p),
$$
where ${ _dC_3}=d(d-1)(d-2)/6$, and the sum is over all the above sets which contain three elements. 
By using $\{a+b\}_+\le\{a\}_++\{b\}_+$ for $a,b\in\re$, one has  
$$
\{Y(p)\}_+\le \frac{1}{ _dC_3}\sum_{\{i,j,k\}}\{Y_{ijk}(p)\}_+.
$$
Therefore, from the monotonicity and convexity of the function $F$, we have 
\begin{equation}
F(\{Y(p)\}_+)\le F\Bigl(\frac{1}{ _dC_3}\sum_{\{i,j,k\}}\{Y_{ijk}(p)\}_+\Bigr)
\le\frac{1}{ _dC_3}\sum_{\{i,j,k\}}F(\{Y_{ijk}(p)\}_+). 
\end{equation}
By integrating the both sides over $p$, we obtain the desired result $(I_d)^2\le (I_3)^2$ for all $d\ge 3$.

\Section{Reflection positivity -- Spin space}

The present system has also a reflection positivity on the spin degrees of freedom \cite{Lieb,KuboKishi}. 
This property yields some inequalities about correlations. 
For example, see (\ref{supercorbound}) and (\ref{ncorbound}) below. 

Recall the expression (\ref{Hint}) of the interaction Hamiltonian   
$$
H_{\rm int}^{(\Lambda)}=g\sum_{|x-y|=1} (a_{x,\uparrow}^\dagger a_{x,\downarrow}^\dagger a_{y,\downarrow}a_{y,\uparrow}
+a_{y,\uparrow}^\dagger a_{y,\downarrow}^\dagger a_{x,\downarrow}a_{x,\uparrow})
$$
in the case of $g'=0$, i.e., with no Coulomb repulsion. Clearly, the two terms of the summand can be written as 
$$
a_{x,\uparrow}^\dagger a_{x,\downarrow}^\dagger a_{y,\downarrow}a_{y,\uparrow}
=a_{x,\uparrow}^\dagger a_{y,\uparrow}\times a_{x,\downarrow}^\dagger a_{y,\downarrow}
$$
and 
$$
a_{y,\uparrow}^\dagger a_{y,\downarrow}^\dagger a_{x,\downarrow}a_{x,\uparrow}
=
a_{y,\uparrow}^\dagger a_{x,\uparrow}\times a_{y,\downarrow}^\dagger a_{x,\downarrow}.
$$
We define an anti-linear map $\vartheta_{\rm spin}$ by 
\begin{equation}
\vartheta_{\rm spin}(a_{x,\uparrow})=a_{x,\downarrow}
\end{equation}
for $x\in\Lambda$. Then, the above two terms can be written 
\begin{equation}
a_{x,\uparrow}^\dagger a_{x,\downarrow}^\dagger a_{y,\downarrow}a_{y,\uparrow}
=a_{x,\uparrow}^\dagger a_{y,\uparrow}\times\vartheta_{\rm spin}(a_{x,\uparrow}^\dagger a_{y,\uparrow})
\end{equation}
and 
\begin{equation}
a_{y,\uparrow}^\dagger a_{y,\downarrow}^\dagger a_{x,\downarrow}a_{x,\uparrow}
=
a_{y,\uparrow}^\dagger a_{x,\uparrow}\times\vartheta_{\rm spin}(a_{y,\uparrow}^\dagger a_{x,\uparrow}).
\end{equation}

Due to a technical reason, we introduce a transformation $a_{x,\sigma}\rightarrow ia_{x,\sigma}$ for 
$x\in\Lambda_{\rm odd}$, $\sigma=\uparrow,\downarrow$. The unitary operator is given by 
\begin{equation}
U_{{\rm odd},\pi/2}^{(\Lambda)}=\prod_{x\in\Lambda_{\rm odd},\sigma=\uparrow,\downarrow}e^{(i\pi/2)n_{x,\sigma}}.
\end{equation}
One has 
\begin{multline}
\check{H}_{\rm int}^{(\Lambda)}:=(U_{{\rm odd},\pi/2}^{(\Lambda)})^\dagger H_{\rm int}^{(\Lambda)}
U_{{\rm odd},\pi/2}^{(\Lambda)}
=-g\sum_{|x-y|=1} (a_{x,\uparrow}^\dagger a_{x,\downarrow}^\dagger a_{y,\downarrow}a_{y,\uparrow}
+a_{y,\uparrow}^\dagger a_{y,\downarrow}^\dagger a_{x,\downarrow}a_{x,\uparrow})\\
=-g\sum_{|x-y|=1} [a_{x,\uparrow}^\dagger a_{y,\uparrow}\vartheta_{\rm spin}(a_{x,\uparrow}^\dagger a_{y,\uparrow})
+a_{y,\uparrow}^\dagger a_{x,\uparrow}\vartheta_{\rm spin}(a_{y,\uparrow}^\dagger a_{x,\uparrow})].
\end{multline}
For each hopping term in the hoping Hamiltonian $H_{\rm hop}^{(\Lambda)}$ of (\ref{Hhop}), we have 
$$
(U_{{\rm odd},\pi/2}^{(\Lambda)})^\dagger i(a_{x,\sigma}^\dagger a_{y,\sigma}-a_{y,\sigma}^\dagger a_{x,\sigma})
U_{{\rm odd},\pi/2}^{(\Lambda)}=
\begin{cases}
a_{x,\sigma}^\dagger a_{y,\sigma}+a_{y,\sigma}^\dagger a_{x,\sigma} & \mbox{for \ } x\in\Lambda_{\rm odd};\\
-(a_{x,\sigma}^\dagger a_{y,\sigma}+a_{y,\sigma}^\dagger a_{x,\sigma}) & \mbox{for \ } y\in\Lambda_{\rm odd},
\end{cases}
$$
where $x,y$ satisfy $|x-y|=1$. We write 
$$
\check{H}_{\rm hop}^{(\Lambda)}:=(U_{{\rm odd},\pi/2}^{(\Lambda)})^\dagger H_{\rm hop}U_{{\rm odd},\pi/2}^{(\Lambda)}. 
$$
Then, the Hamiltonian $\check{H}_{\rm hop}^{(\Lambda)}$ can be decomposed into two parts, 
$$
\check{H}_{\rm hop}^{(\Lambda)}=\check{H}_{{\rm hop},\uparrow}^{(\Lambda)}+\check{H}_{{\rm hop},\downarrow}^{(\Lambda)},
$$
where the two terms in the right-hand side satisfy 
$$
\vartheta_{\rm spin}(\check{H}_{{\rm hop},\downarrow}^{(\Lambda)})=\check{H}_{{\rm hop},\uparrow}^{(\Lambda)}.
$$

The term of the symmetry-breaking field is given by 
$$
H_{\rm SBF}^{(\Lambda)}(B)=-B\sum_{x\in\Lambda} (-1)^{x^{(1)}+\cdots+x^{(d)}}i(a_{x,\uparrow}^\dagger a_{x,\downarrow}^\dagger 
-a_{x,\downarrow}a_{x,\uparrow}).
$$
{from} (\ref{HamB}) and (\ref{O}). One has 
\begin{eqnarray}
\check{H}_{\rm SBF}^{(\Lambda)}(B)
:=(U_{{\rm odd},\pi/2}^{(\Lambda)})^\dagger H_{\rm SBF}^{(\Lambda)}(B)U_{{\rm odd},\pi/2}^{(\Lambda)}
&=&
-B\sum_{x\in\Lambda} i(a_{x,\uparrow}^\dagger a_{x,\downarrow}^\dagger +a_{x,\uparrow}a_{x,\downarrow})\ret
&=&-B\sum_{x\in\Lambda} i[a_{x,\uparrow}^\dagger \vartheta_{\rm spin}(a_{x,\uparrow}^\dagger)
+a_{x,\uparrow}\vartheta_{\rm spin}(a_{x,\uparrow})].\ret
\end{eqnarray}

We write  
$$
\check{H}^{(\Lambda)}(B):=\check{H}_{\rm hop}^{(\Lambda)}+\check{H}_{\rm int}^{(\Lambda)}+\check{H}_{\rm SBF}^{(\Lambda)}(B).
$$
Similarly to the case of the real space, we define  
\begin{eqnarray}
I_{M, {\rm spin}}&:=&\left\{1+\frac{\beta B}{M}
\sum_{x\in\Lambda} i[a_{x,\uparrow}^\dagger \vartheta_{\rm spin}(a_{x,\uparrow}^\dagger) 
+a_{x,\uparrow}\vartheta_{\rm spin}(a_{x,\uparrow})] \right\}\ret
& &\times\left\{1+\frac{\beta g}{M}  
\sum_{|x-y|=1} [a_{x,\uparrow}^\dagger a_{y,\uparrow}\vartheta_{\rm spin}(a_{x,\uparrow}^\dagger a_{y,\uparrow})
+a_{y,\uparrow}^\dagger a_{x,\uparrow}\vartheta_{\rm spin}(a_{y,\uparrow}^\dagger a_{x,\uparrow})]\right\}\ret
& &\times \exp\Bigl[-\frac{\beta}{M}\check{H}_{{\rm hop},\uparrow}^{(\Lambda)}\Bigr]
\exp\Bigl[-\frac{\beta}{M}\vartheta_{\rm spin}(\check{H}_{{\rm hop},\uparrow}^{(\Lambda)})\Bigr].
\end{eqnarray}
Then, 
$$
e^{-\beta\check{H}^{(\Lambda)}}=\lim_{M\nearrow\infty}(I_{M, {\rm spin}})^M. 
$$
The argument in the proof of Proposition~\ref{pro:tildeHBh1d} yields 
$$
{\rm Tr}\; \mathcal{A}_\uparrow \vartheta_{\rm spin}(\mathcal{A}_\uparrow)e^{-\beta\check{H}^{(\Lambda)}(B)}\ge 0
$$
for any operator $\mathcal{A}_\uparrow$ which has even fermion parity and is written in terms of 
the fermion operators $a_{x,\uparrow}$ and $a_{x,\uparrow}^\dagger$ with only the up spin $\uparrow$. 
Therefore, we have 
\begin{equation}
\langle U_{{\rm odd},\pi/2}^{(\Lambda)}\mathcal{A}_\uparrow \vartheta_{\rm spin}(\mathcal{A}_\uparrow)
(U_{{\rm odd},\pi/2}^{(\Lambda)})^\dagger\rangle_{\beta,B}^{(\Lambda)}\ge 0.
\end{equation}
For example, one has 
\begin{equation}
\label{supercorbound}
(-1)^{x^{(1)}+\cdots+x^{(d)}+y^{(1)}+\cdots+y^{(d)}}\langle a_{x,\uparrow}^\dagger a_{x,\downarrow}^\dagger 
a_{y,\downarrow}a_{y,\uparrow}\rangle_{\beta,B}^{(\Lambda)}\ge 0
\end{equation}
for any $x,y\in\Lambda$, and 
\begin{equation}
\label{ncorbound}
\langle (n_{x,\uparrow}-\zeta)(n_{x,\downarrow}-\overline{\zeta})\rangle_{\beta,B}^{(\Lambda)}\ge 0
\end{equation}
for any complex number $\zeta$. 
 
Let us examine the second inequality (\ref{ncorbound}) in more detail. 
Consider the particle-hole transformation, 
\begin{equation}
\label{PHtrans}
U_{\rm PH}^{(\Lambda)}:=\prod_{\sigma=\uparrow,\downarrow}\prod_{x\in\Lambda}u_{x,\sigma},
\end{equation}
where $u_{x,\sigma}$ is given by (\ref{uxsigma}). Then, one has 
$$
(U_{\rm PH}^{(\Lambda)})^\dagger a_{x,\sigma} U_{\rm PH}^{(\Lambda)}=a_{x,\sigma}^\dagger 
$$
for all $x\in\Lambda$ and all $\sigma=\uparrow,\downarrow$. 

First, we show that the Hamiltonian $H^{(\Lambda)}(B)$ of (\ref{HamB}) is invariant under this transformation 
$U_{\rm PH}^{(\Lambda)}$. Actually, one has 
$$
(U_{\rm PH}^{(\Lambda)})^\dagger(a_{x,\sigma}^\dagger a_{y,\sigma}-a_{y,\sigma}^\dagger a_{x,\sigma})U_{\rm PH}^{(\Lambda)}
=a_{x,\sigma} a_{y,\sigma}^\dagger-a_{y,\sigma} a_{x,\sigma}^\dagger
=a_{x,\sigma}^\dagger a_{y,\sigma}-a_{y,\sigma}^\dagger a_{x,\sigma}
$$ 
for the hopping term, 
\begin{eqnarray*}
(U_{\rm PH}^{(\Lambda)})^\dagger(a_{x,\uparrow}^\dagger a_{x,\downarrow}^\dagger a_{y,\downarrow}a_{y,\uparrow}
+a_{y,\uparrow}^\dagger a_{y,\downarrow}^\dagger a_{x,\downarrow}a_{x,\uparrow})U_{\rm PH}^{(\Lambda)}
&=&a_{x,\uparrow} a_{x,\downarrow} a_{y,\downarrow}^\dagger a_{y,\uparrow}^\dagger
+a_{y,\uparrow} a_{y,\downarrow} a_{x,\downarrow}^\dagger a_{x,\uparrow}^\dagger\\
&=&a_{x,\uparrow}^\dagger a_{x,\downarrow}^\dagger a_{y,\downarrow}a_{y,\uparrow}
+a_{y,\uparrow}^\dagger a_{y,\downarrow}^\dagger a_{x,\downarrow}a_{x,\uparrow}
\end{eqnarray*}
for the interaction term, and  
$$
(U_{\rm PH}^{(\Lambda)})^\dagger \Gamma_x^{(2)}U_{\rm PH}^{(\Lambda)}
=(U_{\rm PH}^{(\Lambda)})^\dagger i(a_{x,\uparrow}^\dagger a_{x,\downarrow}^\dagger - a_{x,\downarrow}a_{x,\uparrow})
U_{\rm PH}^{(\Lambda)}
=i(a_{x,\uparrow}a_{x,\downarrow} - a_{x,\downarrow}^\dagger a_{x,\uparrow}^\dagger)=\Gamma_x^{(2)}. 
$$
{From} this invariance of the Hamiltonian, one has 
$$
\langle n_{x,\sigma}\rangle_{\beta,B}^{(\Lambda)}
=\langle (U_{\rm PH}^{(\Lambda)})^\dagger n_{x,\sigma}U_{\rm PH}^{(\Lambda)}\rangle_{\beta,B}^{(\Lambda)}
=\langle (1-n_{x,\sigma})\rangle_{\beta,B}^{(\Lambda)}. 
$$
This implies 
\begin{equation}
\langle n_{x,\sigma}\rangle_{\beta,B}^{(\Lambda)}=\frac{1}{2}
\end{equation}
for all $x\in\Lambda$ and all $\sigma=\uparrow,\downarrow$. 

By choosing $\zeta=\langle n_{x,\uparrow}\rangle_{\beta,B}^{(\Lambda)}=\langle n_{x,\downarrow}\rangle_{\beta,B}^{(\Lambda)}=1/2$ 
in the inequality (\ref{ncorbound}), we obtain 
\begin{equation}
\langle n_{x,\uparrow}n_{x,\downarrow}\rangle_{\beta,B}^{(\Lambda)}\ge 
\langle n_{x,\uparrow}\rangle_{\beta,B}^{(\Lambda)}\langle n_{x,\downarrow}\rangle_{\beta,B}^{(\Lambda)}=\frac{1}{4}.
\end{equation}
Further, by combining these observations with 
the expression (\ref{GammaSq}), we have 
\begin{equation}
\label{Gamma2Lbound}
\langle [\Gamma_x^{(1)}]^2\rangle_{\beta,B}^{(\Lambda)}\ge \frac{1}{2}.
\end{equation}

As an application of the inequality (\ref{Gamma2Lbound}), let us try to use the sum rule \cite{DLS}, 
\begin{equation}
\label{sumruleO}
\frac{1}{|\Lambda|}\sum_p \langle \hat{\Gamma}_p^{(1)}\hat{\Gamma}_{-p}^{(1)}\rangle_{\beta,0}^{(\Lambda)}
=\frac{1}{|\Lambda|}\sum_{x\in\Lambda}\langle [\Gamma_x^{(1)}]^2\rangle_{\beta,0}^{(\Lambda)}, 
\end{equation}
instead of the sum rule (\ref{sumruleKLS}). Although this approach cannot prove the existence of the long-range order 
in the present situation, it is instructive for future studies. In addition, the difference between the spin and fermion systems 
is clarified. 

By using the inequality (\ref{DLSboundO}), one has 
\begin{eqnarray}
\frac{1}{|\Lambda|}\sum_p \langle \hat{\Gamma}_p^{(1)}\hat{\Gamma}_{-p}^{(1)}\rangle_{\beta,0}^{(\Lambda)}
&\le& \frac{1}{|\Lambda|}\sum_{p\ne Q}\frac{1}{2\beta g E_{p+Q}}  
+\frac{1}{|\Lambda|}\sum_{p\ne Q}\sqrt{\frac{\mathfrak{c}_p}{8dgE_{p+Q}}} 
+\frac{1}{|\Lambda|}\langle \hat{\Gamma}_Q^{(1)}\hat{\Gamma}_{Q}^{(1)}\rangle_{\beta,0}^{(\Lambda)}\ret
&\le&\delta'(\beta)+\sqrt{\frac{1}{|\Lambda|}\sum_{p\ne Q}\frac{1}{8dgE_{p+Q}}}
\times \sqrt{\frac{1}{|\Lambda|}\sum_p\mathfrak{c}_p}+(m_{\rm LRO}^{(\Lambda)})^2, 
\end{eqnarray}
where we have used the Schwarz inequality and the positivity of $\mathfrak{c}_p$, and we have written 
$$
\delta'(\beta):=\frac{1}{|\Lambda|}\sum_{p\ne Q}\frac{1}{2\beta g E_{p+Q}}. 
$$
Further, by using the bound (\ref{sumcp}) for $\mathfrak{c}_p$ and the inequality $\sqrt{a+b}\le \sqrt{a}+\sqrt{b}$ 
for $a,b\ge 0$, one has 
$$
\frac{1}{|\Lambda|}\sum_p \langle \hat{\Gamma}_p^{(1)}\hat{\Gamma}_{-p}^{(1)}\rangle_{\beta,0}^{(\Lambda)}
\le \delta'(\beta)+\left[G_d^{(\Lambda)}\right]^{1/2}\times \left[\sqrt{{|\kappa|}/{g}}
+\sqrt{\mathcal{E}_1^{(\Lambda)}/2}\right]+(m_{\rm LRO}^{(\Lambda)})^2,
$$
where 
$$
G_d^{(\Lambda)}:=\frac{1}{|\Lambda|}\sum_{p\ne Q}\sqrt{\frac{1}{E_{p+Q}}}.
$$
Since the interaction Hamiltonian of the present system is very similar to the quantum XY model, 
we can expect that the upper bound $\mathcal{E}_1^{(\Lambda)}\le [d(d+1)]^{1/2}/2$ for the mean energy $\mathcal{E}_1^{(\Lambda)}$ 
in the case of XY model \cite{Anderson,DLS} holds also for the present case.  
If we use this upper bound for $\mathcal{E}_1^{(\Lambda)}$, we have 
$$
\frac{1}{2}\le \frac{1}{|\Lambda|}\sum_p \langle \hat{\Gamma}_p^{(1)}\Gamma_{-p}^{(1)}\rangle_{\beta,0}^{(\Lambda)}
\le \delta''(\beta,|\kappa|/g)+\frac{1}{2}\left\{G_d^{(\Lambda)}[d(d+1)]^{1/2}\right\}^{1/2}
+(m_{\rm LRO}^{(\Lambda)})^2,
$$
where $\delta''(\beta,|\kappa|/g)$ is the small correction for a large $\beta$ and a small $|\kappa|/g$, and  
we have also used (\ref{Gamma2Lbound}) and (\ref{sumruleO}). 
However, the inequality $\lim_{\Lambda\nearrow \ze^d}d G_d^{(\Lambda)}\ge 1$ by \cite{DLS} implies that     
we cannot conclude the existence of the long-range order from the above inequality. 
In the case of the XY model \cite{DLS}, the value $1/2$ in the left-hand side is replaced by $1$, 
and the existence of the long-range order was already proved in \cite{DLS}.

\Section{Nambu-Goldstone modes}

Following \cite{Martin,Momoi,Koma1,Koma3}, we prove the existence of the Nambu-Goldstone mode in this section. 

First, we introduce another interaction,
$$
{H}_{\rm repul}^{(\Lambda)}=g'\sum_{\{x,y\}\subset\Lambda:|x-y|=1}
(n_{x,\uparrow}+n_{x,\downarrow}-1)(n_{y,\uparrow}+n_{y,\downarrow}-1), 
$$
with the coupling constant $g'>0$. This is a nearest-neighbour repulsive Coulomb interaction. For a technical reason, 
we need this interaction. More precisely, we need a bound (\ref{IRBomega}) below for the susceptibility \cite{Momoi} 
for the generator $\Gamma_x^{(3)}$ of the U(1) rotation. This bound is obtained from 
the Gaussian domination bound (\ref{GDBrepul}) below which is 
related to the repulsive Coulomb interaction. Therefore, we assume the coupling constant $g'>0$. 

Clearly, this interaction can be written  
\begin{eqnarray}
H_{\rm repul}^{(\Lambda)}&=&g'\sum_{\{x,y\}\subset\Lambda:|x-y|=1} \Gamma_x^{(3)}\Gamma_y^{(3)} \ret
&=&\frac{g'}{2}\sum_{\{x,y\}\subset\Lambda:|x-y|=1}[\Gamma_x^{(3)}+\Gamma_y^{(3)}]^2 -dg'\sum_{x\in\Lambda} [\Gamma_x^{(3)}]^2 
\end{eqnarray}
in terms of $\Gamma_x^{(3)}=1-n_{x,\uparrow}-n_{x,\downarrow}$. 
Similarly to the case in Sec.~\ref{Sec:RPRS}, we introduce $d$ real-valued functions $h'_m(x)$ on $\Lambda$, and define 
\begin{equation}
\label{Hreplgphp}
H_{\rm repul}^{(\Lambda)}(g',h'):=
\frac{g'}{2}\sum_{x\in\Lambda}\sum_{m=1}^d[\Gamma_x^{(3)}+\Gamma_{x+e_m}^{(3)}+h_m'(x)]^2 
-dg'\sum_{x\in\Lambda} [\Gamma_x^{(3)}]^2.
\end{equation}
We write 
\begin{equation}
H^{(\Lambda)}(B,0,g',h'):=H^{(\Lambda)}(B)+H_{\rm repul}^{(\Lambda)}(g',h'), 
\end{equation}
where $H^{(\Lambda)}(B)$ is the original Hamiltonian given by (\ref{HamB}). 
When $g'/g$ is small, the corresponding contribution by this new interaction 
which appears in the double commutator in $\mathfrak{c}_p$ is small. Therefore, the long-range order still exists, 
and the spontaneous symmetry breaking stably occurs against the weak perturbation.  

The same method as in Sec.~\ref{Sec:RPRS} is applicable again, and one has 
\begin{equation}
\label{GDBrepul}
{\rm Tr}\exp[-\beta H^{(\Lambda)}(B,0,g',h')]\le {\rm Tr}\exp[-\beta H^{(\Lambda)}(B,0,g',0)]. 
\end{equation}
The functions $h'$ are contained only in the Hamiltonian $H_{\rm repul}^{(\Lambda)}(g',h')$ of (\ref{Hreplgphp}).  
The summand of the first sum in the right-hand side of (\ref{Hreplgphp}) is written 
$$
[\Gamma_x^{(3)}+\Gamma_{x+e_m}^{(3)}+h_m'(x)]^2=[\Gamma_x^{(3)}+\Gamma_{x+e_m}^{(3)}]^2 
+2[\Gamma_x^{(3)}+\Gamma_{x+e_m}^{(3)}]h_m'(x)+[h_m'(x)]^2.
$$
For this second term in the right-hand side, one has  
$$
\sum_{x\in\Lambda}\sum_{m=1}^d [\Gamma_x^{(3)}+\Gamma_{x+e_m}^{(3)}]h_m'(x)
=\sum_{x\in\Lambda}\sum_{m=1}^d \Gamma_x^{(3)}[h_m'(x)+h_m'(x-e_m)].
$$
Therefore, one has
\begin{equation} 
\label{IRBrepul}
(\Gamma^{(3)}[\tilde{h}'],\Gamma^{(3)}[\tilde{h}'])_{\beta,B}^{(\Lambda)}(g')\le 
\frac{1}{2\beta g'}\sum_{x\in\Lambda}\sum_{m=1}^d |h_m'(x)|^2
\end{equation}
in the same way as in Sec.~\ref{Sec:RPRS}, where $(\cdots,\cdots)_{\beta,B}^{(\Lambda)}(g')$ is the Duhamel two-point function, 
the function $\tilde{h}'$ is given by 
\begin{equation}
\label{hp}
\tilde{h}'(x):=\sum_{m=1}^d [h_m'(x)+h_m'(x-e_m)], 
\end{equation}
and 
\begin{equation}
\label{Gammatildehp}
\Gamma^{(3)}[\tilde{h'}]=\sum_{x\in\Lambda}\Gamma_x^{(3)}\tilde{h}'(x).
\end{equation}

Note that 
$$
(U_{\rm PH}^{(\Lambda)})^\dagger H^{(\Lambda)}(B,0,g',0)U_{\rm PH}^{(\Lambda)}=H^{(\Lambda)}(B,0,g',0)
$$
and 
$$
(U_{\rm PH}^{(\Lambda)})^\dagger \Gamma_x^{(3)}U_{\rm PH}^{(\Lambda)}=-\Gamma_x^{(3)}
$$
for the unitary operator $U_{\rm PH}^{(\Lambda)}$ of (\ref{PHtrans}). These imply 
\begin{equation}
\label{GammaExpcVal}
\langle\Gamma_x^{(3)}\rangle_{\beta,B}^{(\Lambda)}(g')=0
\end{equation}
where 
$$
\langle\cdots\rangle_{\beta,B}^{(\Lambda)}(g'):=\frac{1}{Z_{\beta,B}^{(\Lambda)}(g')}
{\rm Tr}\;(\cdots) e^{-\beta H^{(\Lambda)}(B,0,g',0)}
$$
with the partition function $Z_{\beta,B}^{(\Lambda)}(g'):={\rm Tr}\;e^{-\beta H^{(\Lambda)}(B,0,g',0)}$. 
We write  
\begin{equation}
\label{omegaBgp}
\omega_{B,g'}^{(\Lambda)}(\cdots):=\frac{1}{q^{(\Lambda)}}\sum_{\mu=1}^{q^{(\Lambda)}}
\langle\Phi_{0,\mu}^{(\Lambda)}(B,g'),(\cdots)\Phi_{0,\mu}^{(\Lambda)}(B,g')\rangle
\end{equation}
for the expectation value in the ground-state sector, where $\Phi_{0,\mu}^{(\Lambda)}(B,g')$, $\mu=1,2,\ldots,q^{(\Lambda)}$, 
are the ground-state eigenvectors of the Hamiltonian $H^{(\Lambda)}(B,0,g',0)$, and $q^{(\Lambda)}$ is the degeneracy.  
Then,   
\begin{equation}
\lim_{\beta\rightarrow\infty}\langle\cdots\rangle_{\beta,B}^{(\Lambda)}(g')=\omega_{B,g'}^{(\Lambda)}(\cdots). 
\end{equation}
In particular, 
\begin{equation}
\omega_{B,g'}^{(\Lambda)}(\Gamma_x^{(3)})=0
\end{equation}
{from} (\ref{GammaExpcVal}). This yields 
\begin{equation}
\label{IRBomega}
\lim_{\beta\rightarrow\infty}
\frac{\beta}{2} (\Gamma^{(3)}[\tilde{h}'],\Gamma^{(3)}[\tilde{h}'])_{\beta,B}^{(\Lambda)}(g')
=\omega_{B,g'}^{(\Lambda)}\bigl(\Gamma^{(3)}[\tilde{h}']\frac{1-P_0^{(\Lambda)}}{\mathcal{H}^{(\Lambda)}(B,g')}
\Gamma^{(3)}[\tilde{h}']\bigr)
\le\frac{1}{4 g'}\sum_{x\in\Lambda}\sum_{m=1}^d |h_m'(x)|^2,
\end{equation}
where we have used the bound (\ref{IRBrepul}), and written 
\begin{equation}
\label{diffHE0}
\mathcal{H}^{(\Lambda)}(B,g'):=H^{(\Lambda)}(B,0,g',0)-E_0^{(\Lambda)}(B,g')
\end{equation} 
with the ground-state energy $E_0^{(\Lambda)}(B,g')$, and $P_0^{(\Lambda)}$ is the spectral projection onto 
the ground-state sector.   

Consider a local order parameter, 
\begin{equation}
\mathcal{A}_R:=\frac{1}{|\Omega_R|}\sum_{x\in\Omega_R}(-1)^{x^{(1)}+\cdots+x^{(d)}}\Gamma_x^{(1)},\quad \mbox{for 
a positive integer\ } R, 
\end{equation}
where 
\begin{equation}
\Omega_R:=\{x\in\ze^d\; |\; |x|_\infty\le R\}
\end{equation}
with the norm $|x|_\infty:=\max_{1\le i\le d}\{|x^{(i)}|\}$ for $x\in\ze^d$. 
By acting the local operator $\mathcal{A}_R$ 
on the ground state $\omega_{B,g'}^{(\Lambda)}(\cdots)$ of (\ref{omegaBgp}), we can construct an excited state above the ground state. 
Following \cite{Koma1}, we consider a trial state, 
\begin{equation}
\label{triallowenergy}
\varphi_{B,\epsilon,R,g'}^{(\Lambda)}(\cdots)
:=\frac{\omega_{B,g'}^{(\Lambda)}(\mathcal{A}_R^\dagger[\mathcal{H}^{(\Lambda)}(B,g')]^{\epsilon/2}(\cdots)
[\mathcal{H}^{(\Lambda)}(B,g')]^{\epsilon/2}
\mathcal{A}_R)}{\omega_{B,g'}^{(\Lambda)}(\mathcal{A}_R^\dagger[\mathcal{H}^{(\Lambda)}(B,g')]^\epsilon\mathcal{A}_R)},
\end{equation}
where $\mathcal{H}^{(\Lambda)}(B,g')$ is given by (\ref{diffHE0}). 
When an excitation energy appeared in this state is very close to the ground-state energy $E_0^{(\Lambda)}(B,g')$, 
the contribution becomes very small due to the factor $[\mathcal{H}^{(\Lambda)}(B,g')]^{\epsilon/2}$. 
Thus, we can eliminate the contributions of the undesired low-lying eigenstates which yield a set of 
ground states in the infinite-volume limit. 
The energy expectation value for the trial state is given by  
\begin{equation}
\label{lowenergy}
\varphi_{B,\epsilon,R,g'}^{(\Lambda)}(\mathcal{H}^{(\Lambda)}(B,g'))
:=\frac{\omega_{B,g'}^{(\Lambda)}(\mathcal{A}_R^\dagger[\mathcal{H}^{(\Lambda)}(B,g')]^{1+\epsilon}
\mathcal{A}_R)}{\omega_{B,g'}^{(\Lambda)}(\mathcal{A}_R^\dagger[\mathcal{H}^{(\Lambda)}(B,g')]^\epsilon\mathcal{A}_R)}.
\end{equation}

\subsection{Estimate of the denominator of the right-hand side of (\ref{lowenergy})}

In order to estimate the denominator of the right-hand side in (\ref{lowenergy}), 
we use the following Kennedy-Lieb-Shastry type inequality \cite{KLS,Koma1}:

\begin{lemma}
\label{lem:BogolyIneq}
Let $\mathcal{A}, \mathcal{C}$ be operators on $\Lambda$, and let $\epsilon$ be a positive small parameter. 
Then, the following bound is valid: 
\begin{eqnarray}
\label{BogolyIneq}
|\omega_{B,g'}^{(\Lambda)}([\mathcal{C},\mathcal{A}])|^2
&\le& \sqrt{D_{B,g'}^{(\Lambda)}(\mathcal{C})}
\sqrt{\varkappa(\epsilon)\;\omega_{B,g'}^{(\Lambda)}(\{\mathcal{C},\mathcal{C}^\ast\})
+\omega_{B,g'}^{(\Lambda)}([[\mathcal{C}^\ast,H^{(\Lambda)}(B,0,g',0)],\mathcal{C}])}\nonumber \\
&\times& \Bigl\{\omega_{B,g'}^{(\Lambda)}(\mathcal{A}[\mathcal{H}^{(\Lambda)}(B,g')]^\epsilon \mathcal{A}^\ast)
+\omega_{B,g'}^{(\Lambda)}(\mathcal{A}^\ast [\mathcal{H}^{(\Lambda)}(B,g')]^\epsilon \mathcal{A})\Bigr\},
\end{eqnarray}
where  
\begin{equation}
D_{B,g'}^{(\Lambda)}(\mathcal{C}):=
\omega_{B,g'}^{(\Lambda)}(\mathcal{C}P_{\rm ex}^{(\Lambda)}(B)[\mathcal{H}^{(\Lambda)}(B,g')]^{-1}\mathcal{C}^\ast) 
+\omega_{B,g'}^{(\Lambda)}(\mathcal{C}^\ast P_{\rm ex}^{(\Lambda)}(B)[\mathcal{H}^{(\Lambda)}(B,g')]^{-1}\mathcal{C})
\label{tildeD}
\end{equation}
with
\begin{equation} 
P_{\rm ex}^{(\Lambda)}:=1-P_0^{(\Lambda)},
\end{equation}
and $\varkappa(\epsilon)$ is a positive function of the parameter $\epsilon$ such that $\varkappa(\epsilon)\rightarrow 0$ as 
$\epsilon\rightarrow 0$. 
\end{lemma}

\noindent
For the derivation of (\ref{BogolyIneq}), see \cite{Koma1}. 

We choose 
\begin{equation}
\mathcal{C}=\Gamma^{(3)}[\tilde{h}']\quad \mbox{and} \quad 
\mathcal{A}=\mathcal{A}_R. 
\end{equation}
We also choose $h_1'(x)=h'(x)$ and $h_m'(x)=0$ for $m=2,3,\ldots,d$, where \cite{Martin} 
\begin{equation}
h'(x):=
\begin{cases}
1, & \mbox{for \ } x\in\Omega_{R+1};\\ 
1-[|x|_\infty-(R+1)]/R, & \mbox{for \ } x\in\Omega_{2R}\backslash\Omega_{R+1};\\
0, & \mbox{otherwise}. 
\end{cases}
\end{equation}
Then, one has 
$$
\Gamma^{(3)}[\tilde{h}']=\sum_{x\in\Omega_{2R}}\Gamma_x^{(3)}h'(x) +\sum_{x\in\Omega_{2R}}\Gamma_x^{(3)}h'(x-e_1), 
$$
\begin{eqnarray}
\label{magnetization}
[\Gamma^{(3)}[\tilde{h}'],\mathcal{A}_R]&=&\frac{2}{|\Omega_R|}\sum_{x\in\Omega_R}(-1)^{x^{(1)}+\cdots+x^{(d)}}
[\Gamma_x^{(3)},\Gamma_x^{(1)}]\ret
&=&\frac{4i}{|\Omega_R|}\sum_{x\in\Omega_R}(-1)^{x^{(1)}+\cdots+x^{(d)}}\Gamma_x^{(2)}
\end{eqnarray}
and 
\begin{equation}
\label{IRBGamma3}
\omega_{B,g'}^{(\Lambda)}\bigl(\Gamma^{(3)}[\tilde{h}']\frac{1-P_0^{(\Lambda)}}{\mathcal{H}^{(\Lambda)}(B,g')}
\Gamma^{(3)}[\tilde{h}']\bigr)
\le\frac{1}{4 g'}(4R+1)^d,
\end{equation} 
where we have used (\ref{hp}), (\ref{Gammatildehp}) and (\ref{IRBomega}). 

\begin{lemma}
The following bound for the double commutator is valid: 
\begin{equation}
\label{DcommuGammaHbound}
\left\Vert\Bigl[\bigl[\Gamma^{(3)}[\tilde{h}'],H^{(\Lambda)}(B,0,g',0)\bigr],\Gamma^{(3)}[\tilde{h}']\Bigr]\right\Vert 
\le \mathcal{K}_1 R^{d-2}+\mathcal{K}_2|B|R^d,
\end{equation}
where $\mathcal{K}_1$ and $\mathcal{K}_2$ are a positive constant. 
\end{lemma}

\begin{proof}{Proof}
Note that  
\begin{equation}
[a_{x,\sigma}^\dagger a_{y,\sigma},n_{x,\sigma}+n_{y,\sigma}]=0
\end{equation}
and
\begin{equation}
h'(x)n_{x,\uparrow}+h'(y)n_{y,\uparrow}=h'(x)(n_{x,\uparrow}+n_{y,\uparrow})
+(h'(y)-h'(x))n_{y,\uparrow}.
\end{equation}
By using these relations, one has 
\begin{eqnarray}
& &[[h'(x)n_{x,\uparrow}+h'(y)n_{y,\uparrow},a_{x,\uparrow}^\dagger a_{y,\uparrow}],h'(x)n_{x,\uparrow}+h'(y)n_{y,\uparrow}]\ret
&=&[[(h'(y)-h'(x))n_{y,\uparrow},a_{x,\uparrow}^\dagger a_{y,\uparrow}],h'(x)n_{x,\uparrow}+h'(y)n_{y,\uparrow}]\ret
&=&[[(h'(y)-h'(x))n_{y,\uparrow},a_{x,\uparrow}^\dagger a_{y,\uparrow}],(h'(y)-h'(x))n_{y,\uparrow}]\ret
&=&|h'(y)-h'(x)|^2[[n_{y,\uparrow},a_{x,\uparrow}^\dagger a_{y,\uparrow}],n_{y,\uparrow}].
\end{eqnarray}
When $x,y$ satisfy $|x-y|=1$, one has $|h'(y)-h'(x)|\le 1/R$ by definition. Therefore, we have 
\begin{equation}
\left\Vert[[h'(x)n_{x,\uparrow}+h'(y)n_{y,\uparrow},a_{x,\uparrow}^\dagger a_{y,\uparrow}],
h'(x)n_{x,\uparrow}+h'(y)n_{y,\uparrow}]\right\Vert\le \frac{4}{R^2}.
\end{equation}

Note that 
\begin{equation}
\Bigl[\bigl[\Gamma^{(3)}[\tilde{h}'],H_{{\rm hop},\uparrow}^{(\Lambda)}\bigr],\Gamma^{(3)}[\tilde{h}']\Bigl]
=\sum_{x,y\in\Omega_{2R}}\Bigl[\bigl[(h'(x)+h'(x-e_1))n_{x,\uparrow},H_{{\rm hop},\uparrow}^{(\Lambda)}\bigr],
(h'(y)+h'(y-e_1))n_{y,\uparrow}\Bigl].
\end{equation}
Since the Hamiltonian $H_{{\rm hop},\uparrow}^{(\Lambda)}$ is written as the sum of $a_{x,\uparrow}^\dagger a_{y,\uparrow}$, 
it is sufficient to estimate 
\begin{equation}
\Bigl[\bigl[\tilde{h}'(x)n_{x,\uparrow}+\tilde{h}'(y)n_{y,\uparrow},a_{x,\uparrow}^\dagger a_{y,\uparrow}\bigr],
\tilde{h}'(x)n_{x,\uparrow}+\tilde{h}'(y)n_{x,\uparrow}\Bigl],
\end{equation}
where $\tilde{h}'(x)=h'(x)+h'(x-e_1)$. In the same way as the above argument, this can be estimated as 
\begin{equation}
\left\Vert\Bigl[\bigl[\tilde{h}'(x)n_{x,\uparrow}+\tilde{h}'(y)n_{y,\uparrow},a_{x,\uparrow}^\dagger a_{y,\uparrow}\bigr],
\tilde{h}'(x)n_{x,\uparrow}+\tilde{h}'(y)n_{x,\uparrow}\Bigl]\right\Vert
\le \frac{16}{R^2}.
\end{equation}
{From} this estimate, one has   
\begin{equation}
\label{DcommuGammaHhop2}
\left\Vert\Bigl[\bigl[\Gamma^{(3)}[\tilde{h}'],H_{{\rm hop},\sigma}^{(\Lambda)}\bigr],\Gamma^{(3)}[\tilde{h}']\Bigl]\right\Vert
\le \frac{32d|\kappa|(4R+1)^d}{R^2}
\end{equation}
for $\sigma=\uparrow,\downarrow$. 

As for the interaction Hamiltonian $H_{\rm int}^{(\Lambda)}$, we use 
$$
a_{x,\uparrow}^\dagger a_{x,\downarrow}^\dagger a_{y,\downarrow}a_{y,\uparrow}
=a_{x,\uparrow}^\dagger a_{y,\uparrow}\cdot a_{x,\downarrow}^\dagger a_{y,\downarrow}.
$$
Note that 
\begin{eqnarray}
& &\Bigl[\bigl[\Gamma^{(3)}[\tilde{h}'],
a_{x,\uparrow}^\dagger a_{y,\uparrow}\cdot a_{x,\downarrow}^\dagger a_{y,\downarrow}\bigr],\Gamma^{(3)}[\tilde{h}']\Bigl]\ret
&=&\Bigl[\bigl[\Gamma^{(3)}[\tilde{h}'],
a_{x,\uparrow}^\dagger a_{y,\uparrow}\bigr]a_{x,\downarrow}^\dagger a_{y,\downarrow}
+a_{x,\uparrow}^\dagger a_{y,\uparrow}\bigl[\Gamma^{(3)}[\tilde{h}'],a_{x,\downarrow}^\dagger a_{y,\downarrow}\bigr]
,\Gamma^{(3)}[\tilde{h}']\Bigl].
\end{eqnarray}
Therefore, we obtain 
\begin{equation}
\left\Vert\Bigl[\bigl[\Gamma^{(3)}[\tilde{h}'],H_{\rm int}^{(\Lambda)}\bigr],\Gamma^{(3)}[\tilde{h}']\Bigl]\right\Vert
\le \frac{128dg(4R+1)^d}{R^2}. 
\end{equation}
We also have 
\begin{equation}
\left\Vert\Bigl[\bigl[\Gamma^{(3)}[\tilde{h}'],B O^{(\Lambda)}\bigr],\Gamma^{(3)}[\tilde{h}']\Bigl]\right\Vert
\le 8|B|(4R+1)^d. 
\end{equation}
By combining these and (\ref{DcommuGammaHhop2}), we obtain the desired bound (\ref{DcommuGammaHbound}). 
\end{proof}

Substituting (\ref{magnetization}), (\ref{IRBGamma3}) and (\ref{DcommuGammaHbound}) 
into the bound (\ref{BogolyIneq}) with $\mathcal{A}=\mathcal{A}_R$ and $\mathcal{C}=\Gamma^{(3)}[\tilde{h}']$, 
we obtain 
\begin{equation}
\label{denomibound}
\frac{8[m_{\rm s}^{(\Lambda)}(B)]^2}{\sqrt{2(4R+1)^d/g'} 
\sqrt{\mathcal{K}_1R^{d-2}+\mathcal{K}_2|B|R^d+8(4R+1)^{2d}\varkappa(\epsilon)}}\le  
\omega_{B,g'}^{(\Lambda)}(\mathcal{A}_R^\dagger [\mathcal{H}^{(\Lambda)}(B,g')]^\epsilon \mathcal{A}_R),
\end{equation}
where we have written 
\begin{equation}
m_{\rm s}^{(\Lambda)}(B):=\frac{1}{|\Omega_R|}\sum_{x\in\Omega_R}\omega_{B,g'}^{(\Lambda)}(\Gamma_x^{(2)})
(-1)^{x^{(1)}+\cdots+x^{(d)}}. 
\end{equation}
In the double limit $B\searrow 0$ and $\Lambda\nearrow\ze^d$, the spontaneous magnetization $m_{\rm s}$ 
for the superconductivity is given by  
\begin{equation}
m_{\rm s}:=\lim_{B\searrow 0}\lim_{\Lambda\nearrow\ze^d}m_{\rm s}^{(\Lambda)}.
\end{equation}
Note that the translational invariance (\ref{transinvGamma}) of the thermal expectation value of $\Gamma_x^{(2)}$ 
also holds in the present situation $g'\ne 0$. Besides, the existence of the long-range order implies 
a non-vanishing spontaneous magnetization in the infinite-volume limit \cite{KomaTasaki1}. 
Therefore, the spontaneous magnetization $m_{\rm s}$ is strictly positive for a large $R$. 

When we choose $\epsilon=0$, we obtain 
\begin{equation}
\frac{8[m_{\rm s}]^2}{\sqrt{2(4R+1)^d/g'} \sqrt{\mathcal{K}_1R^{d-2}}}\le 
\omega_{0,g'}(\mathcal{A}_R^\dagger \mathcal{A}_R)
\end{equation}
from the above inequality (\ref{denomibound}), where we have written 
\begin{equation}
\omega_{0,g'}(\cdots):={\rm weak}^\ast\mbox{-}\lim_{B\searrow 0}{\rm weak}^\ast\mbox{-}
\lim_{\Lambda\nearrow\ze^d}\omega_{B,g'}^{(\Lambda)}(\cdots). 
\end{equation}
This bound rules out the possibility of the rapid decay $o(|x-y|^{-(d-1)})$ for the correlation 
$\omega_{0,g'}(\Gamma_x^{(1)}\Gamma_y^{(1)})$, where $o(\varepsilon)$ denotes a quantity $q(\varepsilon)$ 
such that $q(\varepsilon)/\varepsilon$ is vanishing in the limit $\varepsilon\searrow 0$.   

In passing, we can obtain a similar bound for the decay of the correlation at non-zero temperatures 
by using Bogoliubov inequality \cite{Bogoliubov,DLS}, 
\begin{equation}
\bigl|\langle [\mathcal{C},\mathcal{A}]\rangle_{\beta,B}^{(\Lambda)}(g')\bigr|^2
\le \frac{\beta}{2}\langle [\mathcal{C}^\dagger,[H^{(\Lambda)}(B,0,g',0),\mathcal{C}]]\rangle_{\beta,B}^{(\Lambda)}(g')
\langle \{\mathcal{A},\mathcal{A}^\dagger\}\rangle_{\beta,B}^{(\Lambda)}(g'), 
\end{equation}
for operators, $\mathcal{C}$ and $\mathcal{A}$. Similarly, from (\ref{magnetization}) and (\ref{DcommuGammaHbound}), we have 
\begin{equation}
\frac{16[m_{\rm s}^{(\Lambda)}(\beta,B,g')]^2}{\mathcal{K}_1R^{d-2}+\mathcal{K}_2|B|R^d}
\le \beta \langle \mathcal{A}_R^\dagger \mathcal{A}_R\rangle_{\beta,B}^{(\Lambda)}(g'),
\end{equation}
where we have written 
$$
m_{\rm s}^{(\Lambda)}(\beta,B,g'):=\frac{1}{|\Omega_R|}\sum_{x\in\Omega_R}
(-1)^{x^{(1)}+\cdots+x^{(d)}}\langle \Gamma_x^{(2)}\rangle_{\beta,B}^{(\Lambda)}(g').
$$
We write 
$$
\langle \cdots \rangle_{\beta}(g'):={\rm weak}^\ast\mbox{-}\lim_{B\searrow 0}
{\rm weak}^\ast\mbox{-}\lim_{\Lambda\nearrow\ze^d}\langle\cdots\rangle_{\beta,B}^{(\Lambda)}(g')
$$
and 
$$
m_{\rm s}(\beta,g'):=\lim_{B\searrow 0}\lim_{\Lambda\nearrow\ze^d}m_{\rm s}^{(\Lambda)}(\beta,B,g')
$$
with the use of the same sequences in the double limit as those in the thermal average. 
Then, we have 
\begin{equation}
\frac{16[m_{\rm s}(\beta,g')]^2}{\mathcal{K}_1R^{d-2}}
\le \beta \langle \mathcal{A}_R^\dagger \mathcal{A}_R\rangle_{\beta}(g').
\end{equation}
This rules out the possibility of the rapid decay $o(|x-y|^{-(d-2)})$ for the transverse correlation 
$\langle \Gamma_x^{(1)}\Gamma_y^{(1)}\rangle_{\beta}(g')$ when the spontaneous magnetization $m_{\rm s}(\beta,g')$ 
of the superconductivity is non-vanishing.

\subsection{Estimate of the numerator of the right-hand side of (\ref{lowenergy})}

Next let us estimate the numerator of (\ref{lowenergy}).
For the Hamiltonian $\mathcal{H}^{(\Lambda)}(B,g')$, we denote by $P(E',+\infty)$ the spectral projection 
onto the energies which are larger than $E'>0$. We also write $P[0,E'):=1-P(E',+\infty)$. 
Note that
\begin{eqnarray*}
& &\omega_B^{(\Lambda)}(\mathcal{A}_R[\mathcal{H}^{(\Lambda)}(B,g')]^{1+\epsilon}\mathcal{A}_R)\\
&=&\omega_B^{(\Lambda)}(\mathcal{A}_RP[0,E')[\mathcal{H}^{(\Lambda)}(B,g')]^{1+\epsilon}\mathcal{A}_R)
+\omega_B^{(\Lambda)}(\mathcal{A}_RP(E',+\infty)[\mathcal{H}^{(\Lambda)}(B,g')]^{1+\epsilon}\mathcal{A}_R)\\
&\le&\omega_B^{(\Lambda)}(\mathcal{A}_RP[0,E')\mathcal{H}^{(\Lambda)}(B,g')\mathcal{A}_R)
\times(E')^{\epsilon}\\
&+&\omega_B^{(\Lambda)}(\mathcal{A}_RP(E',+\infty)[\mathcal{H}^{(\Lambda)}(B,g')]^3\mathcal{A}_R)
\times (E')^{\epsilon-2}\\
&\le&\omega_B^{(\Lambda)}(\mathcal{A}_R\mathcal{H}^{(\Lambda)}(B,g')\mathcal{A}_R)
\times(E')^{\epsilon}
+\omega_B^{(\Lambda)}(\mathcal{A}_R[\mathcal{H}^{(\Lambda)}(B,g')]^3\mathcal{A}_R)\times (E')^{\epsilon-2}.
\end{eqnarray*}
The first term in the right-hand side in the last line can be estimated as 
\begin{equation}
\label{omegaBAHAexpvalue}
\omega_B^{(\Lambda)}(\mathcal{A}_R\mathcal{H}^{(\Lambda)}(B,g')\mathcal{A}_R)
=\frac{1}{2}\omega_B^{(\Lambda)}([\mathcal{A}_R,[H^{(\Lambda)}(B,0,g',0),\mathcal{A}_R]])
\le\frac{{\cal K}_3}{R^d}. 
\end{equation}
with the positive constant ${\cal K}_3$. Similarly, the second term is evaluated as 
\begin{eqnarray}
\label{omegaARPEPH3AR}
& &\omega_B^{(\Lambda)}(\mathcal{A}_R[\mathcal{H}^{(\Lambda)}(B,g')]^3\mathcal{A}_R)\ret
&\le&
\omega_B^{(\Lambda)}([\mathcal{A}_R,H^{(\Lambda)}(B,0,g',0)]\mathcal{H}^{(\Lambda)}(B,g')
[H^{(\Lambda)}(B,0,g',0),\mathcal{A}_R])\nonumber\\
&=&\omega_B^{(\Lambda)}(\mathcal{B}_R\mathcal{H}^{(\Lambda)}(B,g')\mathcal{B}_R)\nonumber\\
&=&\frac{1}{2}\omega_B^{(\Lambda)}([\mathcal{B}_R,[H^{(\Lambda)}(B,0,g',0),\mathcal{B}_R]])
\le\frac{{\cal K}_4}{R^d}
\end{eqnarray}
with the positive constant ${\cal K}_4$, 
where we have written $\mathcal{B}_R:=i[H^{(\Lambda)}(B,0,g',0),\mathcal{A}_R]$, and used the assumption that 
the interactions of the Hamiltonian $H^{(\Lambda)}(B,0,g',0)$ are of finite range.  
{From} these observations, we obtain 
\begin{equation}
\label{finalnumebound}
\omega_B^{(\Lambda)}(\mathcal{A}_R[\mathcal{H}^{(\Lambda)}(B,g')]^{1+\epsilon}\mathcal{A}_R)
\le \frac{1}{R^d}\left[\mathcal{K}_3(E')^\epsilon+\mathcal{K}_4(E')^{\epsilon-2}\right]
\le \frac{{\cal K}_3+{\cal K}_4}{R^d},
\end{equation}
where we have chosen $E'=1$. 

\subsection{A gapless excitation above the ground state}

By using (\ref{denomibound}) and  (\ref{finalnumebound}) for estimating the right-hand side of (\ref{lowenergy}), we obtain 
\begin{equation}
\varphi_{B,\epsilon,R,g'}^{(\Lambda)}(\mathcal{H}^{(\Lambda)}(B,g'))
 \le \frac{\sqrt{2(4R+1)^d}\sqrt{\mathcal{K}_1R^{d-2}+\mathcal{K}_2|B|R^d+8(4R+1)^{2d}\varkappa(\epsilon)}}{8\sqrt{g'}
[m_{\rm s}^{(\Lambda)}(B)]^2}\times \frac{{\cal K}_3+{\cal K}_4}{R^d}. 
\end{equation}
We choose the parameter $\epsilon$ so that it satisfies 
$$
8(4R+1)^{2d}\varkappa(\epsilon)\le R^{d-2}
$$
for a given $R$. Then, the above bound can be written as 
\begin{equation}
\varphi_{B,\epsilon,R,g'}^{(\Lambda)}(\mathcal{H}^{(\Lambda)}(B,g'))
 \le \frac{\sqrt{2(4+1/R)^d}\sqrt{1+\mathcal{K}_1+\mathcal{K}_2|B|R^2}}{8\sqrt{g'}
[m_{\rm s}^{(\Lambda)}(B)]^2}\times \frac{{\cal K}_3+{\cal K}_4}{R}. 
\end{equation}
In the double limit $B\searrow 0$ and $\Lambda\nearrow\ze^d$, one has 
\begin{equation}
\lim_{B\searrow 0}\lim_{\Lambda\nearrow\ze^d} \varphi_{B,\epsilon,R,g'}^{(\Lambda)}(\mathcal{H}^{(\Lambda)}(B,g'))
 \le \frac{\sqrt{2(4+1/R)^d}\sqrt{1+\mathcal{K}_1}}{8\sqrt{g'}
[m_{\rm s}]^2}\times \frac{{\cal K}_3+{\cal K}_4}{R}. 
\end{equation}
Since we can take $R$ to be any large positive integer, this inequality implies 
that there exists a gapless local excitation above the infinite-volume ground state $\omega_{0,g'}(\cdots)$.  

Next, we construct a quasi-local operator which creates a low energy excitation 
above the infinite-volume ground state $\omega_{0,g'}(\cdots)$.  
In the same way as in Sec.~5.3 in \cite{Koma1}, we can find a non-negative real-valued 
function $\hat{\chi}\in C_0^\infty(\re)$ such that 
the function $\hat{\chi}$ satisfies ${\rm supp}\;\hat{\chi}\subseteq (0,\Delta E_1)$ with a constant $\Delta E_1>0$ 
and further satisfies the following two bounds: 
\begin{equation}
\label{lowerboundchinorm}
\frac{|m_{\rm s}^{(\Lambda)}(B)|^2}{\mathcal{K}_0R^{d-1}[1+\mathcal{K}_1+\mathcal{K}_2|B|R^2]^{1/2}}
\le 
\omega_{B,g'}^{(\Lambda)}(\mathcal{A}_R^\dagger [\hat{\chi}(\mathcal{H}^{(\Lambda)}(B,g'))]^2\mathcal{A}_R)
\end{equation}
and 
\begin{equation}
\label{upperboundchiH}
\omega_{B,g'}^{(\Lambda)}(\mathcal{A}_R^\dagger \hat{\chi}(\mathcal{H}^{(\Lambda)}(B,g'))
\mathcal{H}^{(\Lambda)}(B,g')\hat{\chi}(\mathcal{H}^{(\Lambda)}(B,g'))\mathcal{A}_R)\le 
\frac{1}{R^d}(\mathcal{K}_3+\mathcal{K}_4). 
\end{equation}
Here, $\mathcal{K}_0$ is a positive constant, and these two bounds are derived from (\ref{denomibound}) 
and (\ref{finalnumebound}), respectively. 
For a local operator $\mathcal{A}$, we define 
\begin{equation}
\tau_{t,B}^{(\Lambda)}(\mathcal{A}):=\exp[iH^{(\Lambda)}(B,0,g',0)t]\mathcal{A}\exp[-iH^{(\Lambda)}(B,0,g',0)t]
\end{equation}
and 
\begin{equation}
\tau_{*\chi,B}^{(\Lambda)}(\mathcal{A}):=\int_{-\infty}^{+\infty}dt\; \chi(t) \tau_{t,B}^{(\Lambda)}(\mathcal{A}),
\end{equation}
where the function $\chi$ is the Fourier transform of $\hat{\chi}$. Then, the following two limits exist: 
\begin{equation}
\tau_{t,0}(\mathcal{A}):=\lim_{B\searrow 0}\lim_{\Lambda\nearrow\ze^d}\tau_{t,B}^{(\Lambda)}(\mathcal{A})
\end{equation}
and 
\begin{equation}
\tau_{*\chi,0}(\mathcal{A}):=\lim_{B\searrow 0}\lim_{\Lambda\nearrow\ze^d}\tau_{*\chi,B}^{(\Lambda)}(\mathcal{A}).
\end{equation}
For the ground-state vector $\Phi_{0,\mu}^{(\Lambda)}(B,g')$ of the Hamiltonian $H^{(\Lambda)}(B,0,g',0)$, 
one has 
\begin{equation}
\tau_{*\chi,B}^{(\Lambda)}(\mathcal{A})\Phi_{0,\mu}^{(\Lambda)}(B,g')
=\hat{\chi}(\mathcal{H}^{(\Lambda)}(B,g'))\mathcal{A}\Phi_{0,\mu}^{(\Lambda)}(B,g').  
\end{equation}
This implies 
\begin{equation}
\label{omegatauchinorm}
\omega_{B,g'}^{(\Lambda)}(\tau_{*\chi,B}^{(\Lambda)}(\mathcal{A}_R^\dagger)\tau_{*\chi,B}^{(\Lambda)}(\mathcal{A}_R))
=\omega_{B,g'}^{(\Lambda)}(\mathcal{A}_R^\dagger [\hat{\chi}(\mathcal{H}^{(\Lambda)}(B,g')]^2\mathcal{A}_R)
\end{equation}
and 
\begin{eqnarray}
\label{omegatauchiH}
& &\omega_{B,g'}^{(\Lambda)}(\tau_{*\chi,B}^{(\Lambda)}(\mathcal{A}_R^\dagger)\mathcal{H}^{(\Lambda)}(B,g')
\tau_{*\chi,B}^{(\Lambda)}(\mathcal{A}_R))\ret
&=&\omega_{B,g'}^{(\Lambda)}(\mathcal{A}_R^\dagger \hat{\chi}(\mathcal{H}^{(\Lambda)}(B,g'))\mathcal{H}^{(\Lambda)}(B,g')
\hat{\chi}(\mathcal{H}^{(\Lambda)}(B,g'))\mathcal{A}_R).
\end{eqnarray}
Since the function $\hat{\chi}$ satisfies ${\rm supp}\;\hat{\chi}\subseteq (0,\Delta E_1)$, 
the factor $\hat{\chi}(\mathcal{H}^{(\Lambda)}(B,g'))$ yields a projection onto the excited states above the sector of 
the ground states.

The following relation of the excitation energies 
between the infinite-volume and the finite-volume ground states is valid: \cite{AL,Koma,Koma1,Koma3} 
\begin{eqnarray}
& &\lim_{\Lambda'\nearrow\ze^d} \frac{\omega_{0,g'}\bigl([\tau_{*\chi,0}(\mathcal{A}_R)]^\dagger 
[H^{(\Lambda')}(0,0,g',0),\tau_{*\chi,0}(\mathcal{A}_R)]\bigr)}{\omega_{0,g'}\bigl([\tau_{*\chi,0}(\mathcal{A}_R)]^\dagger 
\tau_{*\chi,0}(\mathcal{A}_R)\bigr)}\ret
&=&\lim_{B\searrow 0}\lim_{\Lambda\nearrow \ze^d} 
\frac{\omega_{B,g'}^{(\Lambda)}\bigl([\tau_{*\chi,B}^{(\Lambda)}(\mathcal{A}_R]^\dagger 
[H^{(\Lambda)}(B,0,g',0),\tau_{*\chi,B}^{(\Lambda)}(\mathcal{A}_R)]\bigr)}{\omega_{B,g'}^{(\Lambda)}
\bigl([\tau_{*\chi,B}^{(\Lambda)}(\mathcal{A}_R]^\dagger \tau_{*\chi,B}^{(\Lambda)}(\mathcal{A}_R)\bigr)}.    
\end{eqnarray}
Here, the operator $\tau_{*\chi,0}(\mathcal{A}_R)$ is quasi-local because the function $\chi(t)$ rapidly decays for 
large $|t|$ by definition. 
Combining this, (\ref{lowerboundchinorm}), (\ref{upperboundchiH}), (\ref{omegatauchinorm}) and (\ref{omegatauchiH}), we obtain 
the desired result, 
\begin{eqnarray}
& &\lim_{\Lambda'\nearrow\ze^d} \frac{\omega_{0,g'}\bigl([\tau_{*\chi,0}(\mathcal{A}_R)]^\dagger 
[H^{(\Lambda')}(0,0,g',0),\tau_{*\chi,0}(\mathcal{A}_R)]\bigr)}{\omega_{0,g'}\bigl([\tau_{*\chi,0}(\mathcal{A}_R)]^\dagger 
\tau_{*\chi,0}(\mathcal{A}_R)\bigr)}\ret
&=&\lim_{B\searrow 0}\lim_{\Lambda\nearrow \ze^d} 
\frac{\omega_{B,g'}^{(\Lambda)}\bigl([\tau_{*\chi,B}^{(\Lambda)}(\mathcal{A}_R]^\dagger 
[H^{(\Lambda)}(B,0,g',0),\tau_{*\chi,B}^{(\Lambda)}(\mathcal{A}_R)]\bigr)}{\omega_{B,g'}^{(\Lambda)}
\bigl([\tau_{*\chi,B}^{(\Lambda)}(\mathcal{A}_R]^\dagger \tau_{*\chi,B}^{(\Lambda)}(\mathcal{A}_R)\bigr)}
\le \frac{{\rm Const.}}{m_{\rm s}^2 R}.     
\end{eqnarray}

\appendix

\Section{U(1) rotation of the order parameters}
\label{appendix:U(1)rotation}

Consider the transformation, 
$$
\Gamma_x^{(1)}\rightarrow e^{-i\theta \Gamma_x^{(3)}/2}\Gamma_x^{(1)}e^{i\theta\Gamma_x^{(3)}/2},
$$
where $\theta$ is a real variable. Note that 
\begin{eqnarray*}
\frac{d}{d\theta}e^{-i\theta \Gamma_x^{(3)}/2}\Gamma_x^{(1)}e^{i\theta\Gamma_x^{(3)}/2}&=&
-\frac{i}{2}e^{-i\theta \Gamma_x^{(3)}/2}[\Gamma_x^{(3)},\Gamma_x^{(1)}]e^{i\theta\Gamma_x^{(3)}/2}\ret
&=&e^{-i\theta \Gamma_x^{(3)}/2}\Gamma_x^{(2)}e^{i\theta\Gamma_x^{(3)}/2},
\end{eqnarray*}
where we have used the commutation relation $[\Gamma_x^{(3)},\Gamma_x^{(1)}]=2i\Gamma_x^{(2)}$. 
Further, 
\begin{eqnarray*}
\frac{d^2}{d\theta^2}e^{-i\theta \Gamma_x^{(3)}/2}\Gamma_x^{(1)}e^{i\theta\Gamma_x^{(3)}/2}&=&
-\frac{i}{2}e^{-i\theta \Gamma_x^{(3)}/2}[\Gamma_x^{(3)},\Gamma_x^{(2)}]e^{i\theta\Gamma_x^{(3)}/2}\ret
&=&-e^{-i\theta \Gamma_x^{(3)}/2}\Gamma_x^{(1)}e^{i\theta\Gamma_x^{(3)}/2},
\end{eqnarray*}
where we have used $[\Gamma_x^{(2)},\Gamma_x^{(3)}]=2i\Gamma_x^{(1)}$. The solution of this differential equation 
is given by 
$$
e^{-i\theta \Gamma_x^{(3)}/2}\Gamma_x^{(1)}e^{i\theta\Gamma_x^{(3)}/2}
=\Gamma_x^{(1)}\cos\theta+\Gamma_x^{(2)}\sin\theta.
$$
For $\theta=\pm \pi/2$, one has 
\begin{equation}
\label{U1rotation}
e^{-i\pi \Gamma_x^{(3)}/4}\Gamma_x^{(1)}e^{i\pi\Gamma_x^{(3)}/4}=\Gamma_x^{(2)}
\quad \mbox{and}\quad 
e^{-i\pi \Gamma_x^{(3)}/4}\Gamma_x^{(2)}e^{i\pi\Gamma_x^{(3)}/4}=-\Gamma_x^{(1)}.
\end{equation}
The global rotation is given by 
\begin{equation}
\label{GU1rotation}
U_{\rm rot}^{(\Lambda)}(\theta):=\prod_{x\in\Lambda} e^{i\theta\Gamma_x^{(3)}/2}. 
\end{equation}

\Section{Proof of the inequality (\ref{E1bound}) for $\mathcal{E}_1$}    
\label{Appendix:meanenergy}

The thermal expectation value of the interaction Hamiltonian $H_{\rm int}^{(\Lambda)}$ can be written as 
\begin{eqnarray}
\langle H_{\rm int}^{(\Lambda)}\rangle_{\beta,0}^{(\Lambda)}&=&
\frac{g}{2}\sum_{|x-y|=1}\langle\Gamma_x^{(1)}\Gamma_y^{(1)}+\Gamma_x^{(2)}\Gamma_y^{(2)}\rangle_{\beta,0}^{(\Lambda)}\ret
&=&g\sum_{|x-y|=1}\langle\Gamma_x^{(1)}\Gamma_y^{(1)}\rangle_{\beta,0}^{(\Lambda)}\ret
&=&\frac{g}{2}\sum_{x\in\Lambda}\sum_{m=1}^d [\langle \Gamma_x^{(1)}\Gamma_{x+e_m}^{(1)}\rangle_{\beta,0}^{(\Lambda)}+
\langle\Gamma_x^{(1)}\Gamma_{x-e_m}^{(1)}\rangle_{\beta,0}^{(\Lambda)}]\ret
&=&\frac{dg}{2}\sum_{x\in\Lambda} [\langle \Gamma_x^{(1)}\Gamma_{x+e_1}^{(1)}\rangle_{\beta,0}^{(\Lambda)}+
\langle \Gamma_x^{(1)}\Gamma_{x-e_1}^{(1)}\rangle_{\beta,0}^{(\Lambda)}],
\end{eqnarray}
where we have used (\ref{U1rotation}), (\ref{transinvGammaGamma}) and (\ref{direcIndep}).  
As for the hopping term, one has 
\begin{equation}
-d|\kappa||\Lambda|\le \langle H_{\rm hop}^{(\Lambda)}\rangle_{\beta,0}^{(\Lambda)}. 
\end{equation}
Combining these observations with the expression (\ref{E1}) of $\mathcal{E}_1^{(\Lambda)}$, we have 
\begin{equation}
\label{E1E0bound}
-d|\kappa|\le\mathcal{E}_{\beta,0}^{(\Lambda)}+dg\mathcal{E}_1^{(\Lambda)},
\end{equation}
where we have written 
\begin{equation}
\mathcal{E}_{\beta,0}^{(\Lambda)}:=\frac{1}{|\Lambda|}\langle H^{(\Lambda)}(0)\rangle_{\beta,0}^{(\Lambda)}. 
\end{equation}

Next, we estimate the mean energy $\mathcal{E}_{\beta,0}^{(\Lambda)}$. For this purpose, we write   
\begin{equation}
\mathcal{F}_{\beta,0}^{(\Lambda)}
:=-\frac{1}{\beta|\Lambda|}\log Z_{\beta,0}^{(\Lambda)}
\end{equation}
for the free energy, and   
\begin{equation}
\mathcal{S}_{\beta,0}^{(\Lambda)}:=-\frac{1}{|\Lambda|}{\rm Tr}\; \rho_{\beta,0}^{(\Lambda)}\log \rho_{\beta,0}^{(\Lambda)}
\end{equation}
for the entropy, where $Z_{\beta,0}^{(\Lambda)}={\rm Tr}\; e^{-\beta H^{(\Lambda)}(0)}$ and 
$\rho_{\beta,0}^{(\Lambda)}:=e^{-\beta H^{(\Lambda)}(0)}/Z_{\beta,0}^{(\Lambda)}$. 
As is well known, the following relation holds: 
\begin{equation}
\mathcal{F}_{\beta,0}^{(\Lambda)}=\mathcal{E}_{\beta,0}^{(\Lambda)}-\beta^{-1}\mathcal{S}_{\beta,0}^{(\Lambda)}.
\end{equation}
The free energy in the infinite-volume limit $\Lambda\nearrow \ze^d$ exists and does not depend on the boundary conditions, and 
the following zero temperature limit $\beta\nearrow\infty$ exists and equals the mean energy $\mathcal{E}_0$ 
of a ground state $\omega_0$ of the corresponding translationally 
invariant system:\footnote{See, e.g., Chapter~6.2 of the book \cite{BR}.}  
\begin{equation}
\lim_{\beta\nearrow\infty}\lim_{\Lambda\nearrow\ze^d}\mathcal{F}_{\beta,0}^{(\Lambda)}
=\mathcal{E}_0.
\end{equation}
Since the entropy term $\beta^{-1}\mathcal{S}_{\beta,0}^{(\Lambda)}$ is vanishing in the limit $\beta\nearrow\infty$,   
these observations imply that for any given small $\tilde{\delta}>0$, there exists a sufficiently large $\beta$ such that 
\begin{equation}
\Bigl|\lim_{\Lambda\nearrow\ze^d}\mathcal{E}_{\beta,0}^{(\Lambda)}-\mathcal{E}_0\Bigr|\le \tilde{\delta}. 
\end{equation}
In other words, there exists a positive function $\tilde{\delta}(\beta)$ of $\beta$ such that 
the function $\tilde{\delta}(\beta)$ becomes small for a large $\beta$, and that 
\begin{equation}
\mathcal{E}_0- \tilde{\delta}(\beta)\le 
\lim_{\Lambda\nearrow\ze^d}\mathcal{E}_{\beta,0}^{(\Lambda)}\le \mathcal{E}_0+ \tilde{\delta}(\beta). 
\end{equation}
Combining this with (\ref{E1E0bound}), we obtain 
\begin{equation}
\label{E1E0bound2}
-dg\mathcal{E}_1-d|\kappa|\le \mathcal{E}_0+\tilde{\delta}(\beta),
\end{equation}
where we have written 
\begin{equation}
\mathcal{E}_1:=\lim_{\Lambda\nearrow\ze^d}\mathcal{E}_1^{(\Lambda)}. 
\end{equation}

In order to estimate the mean energy $\mathcal{E}_0$, we consider the present system with the periodic boundary condition. 
We write $H_{\rm P}^{(\Lambda)}(0)$ for the Hamiltonian without the external symmetry-breaking field. 
Consider two state vectors, 
\begin{equation}
(1\pm a_{x,\uparrow}^\dagger a_{x,\downarrow}^\dagger)|0\rangle,  
\end{equation}
where $|0\rangle$ is the vacuum for the fermions, i.e., $a_{x,\sigma}|0\rangle=0$ for all $x\in\Lambda$ 
and $\sigma=\uparrow, \downarrow$. 
Then, one has 
\begin{equation}
\label{eigenGamma1}
(a_{x,\uparrow}^\dagger a_{x,\downarrow}^\dagger + a_{x,\downarrow}a_{x,\uparrow})
(1\pm a_{x,\uparrow}^\dagger a_{x,\downarrow}^\dagger)|0\rangle
=\pm (1\pm a_{x,\uparrow}^\dagger a_{x,\downarrow}^\dagger)|0\rangle.
\end{equation}
Namely, these vectors are eigenstates of the operator $\Gamma_x^{(1)}$.  
Further, 
\begin{equation}
\label{exchangeGamma2}
(a_{x,\uparrow}^\dagger a_{x,\downarrow}^\dagger - a_{x,\downarrow}a_{x,\uparrow})
(1\pm a_{x,\uparrow}^\dagger a_{x,\downarrow}^\dagger)|0\rangle
=\mp (1\mp a_{x,\uparrow}^\dagger a_{x,\downarrow}^\dagger)|0\rangle.
\end{equation}
This implies that the operator $\Gamma_x^{(2)}$ exchanges the two vectors. 
In order to estimate the mean energy $\mathcal{E}_0^{(\Lambda)}$ of the finite-volume ground state $\omega_0^{(\Lambda)}$, 
let us consider a variational state, 
\begin{equation}
\Phi_{\rm var}:=\Biggl[\prod_{x\in\Lambda_{\rm odd}}\frac{1}{\sqrt{2}}(1-a_{x,\uparrow}^\dagger a_{x,\downarrow}^\dagger)\Biggr]
\Biggl[\prod_{x\in\Lambda\backslash\Lambda_{\rm odd}}\frac{1}{\sqrt{2}}
(1+a_{x,\uparrow}^\dagger a_{x,\downarrow}^\dagger)\Biggr]|0\rangle. 
\end{equation}
The variational principle yields 
\begin{equation}
\omega_{0}^{(\Lambda)}(H_{\rm p}^{(\Lambda)}(0))\le 
\langle \Phi_{\rm var},H_{\rm }^{(\Lambda)}(0)\Phi_{\rm var}\rangle=-\frac{dg}{2}|\Lambda|
\end{equation}
because the contribution from the hopping Hamiltonian $H_{\rm hop, P}^{(\Lambda)}$ is vanishing 
for the variational state $\Phi_{\rm var}$. Here, we have also used the expression (\ref{HintGamma12}) and 
the above relations, (\ref{eigenGamma1}) and (\ref{exchangeGamma2}), for getting the estimate in the right-hand side. 
Therefore, we have 
\begin{equation}
\mathcal{E}_0\le -\frac{dg}{2}
\end{equation}
in the infinite-volume limit. 
Substituting this into the right-hand side of (\ref{E1E0bound2}), we obtain the desired bound, 
$$
\mathcal{E}_1\ge \frac{1}{2}-\frac{\tilde{\delta}(\beta)}{dg}-\frac{|\kappa|}{g}.
$$


\end{document}